\documentclass{article}

\usepackage{amsmath,amssymb,amsfonts,mathrsfs,latexsym}
\usepackage{tikz}
\usepackage{booktabs}
\usepackage{multirow}
\usepackage{graphicx}		
\usepackage{natbib,amsthm}
\usepackage{times}
\usepackage{appendix}
\usepackage{enumerate,enumitem}
\usepackage{tikz}
\usetikzlibrary{math}

\newtheorem{theorem}{Theorem}
\newtheorem{proposition}{Proposition}
\newtheorem{lemma}{Lemma}
\newtheorem{remark}{Remark}
\newtheorem{corollary}{Corollary}
\newtheorem{definition}{Definition}

\newcommand{\rcolor}[1]{{#1}}
\newcommand{\displacement}[1]{u^{#1}}
\newcommand{\ii}{\mathrm{i}}
\newcommand{\dd}{\mathrm{d}}

\renewcommand{\dd}{\,\operatorname{d}}
\renewcommand{\ii}{\mathrm{i}}
\newcommand{\Id}{\operatorname{Id}}

\newcommand{\mb}[1]{\ensuremath{\mathbb{#1}}}

\newcommand{\RR}{\mb{R}}

\newcommand{\C}{\mb{C}}

\def\vp{\varphi}

\DeclareMathOperator{\Div}{Div}

\newcommand{\tred}[1]{#1}

\title{Characterization of the spectra of rotating truncated gas planets and inertia-gravity modes}

\author{Maarten V. de Hoop, Sean Holman, Alexei Iantchenko\footnote{We dedicate this paper to the memory of our friend and collaborator Alexei Iantchenko who tragically passed away while working on this research.}}

\date{}

\begin{document}

\maketitle

\begin{abstract}
    We study the essential spectrum, which corresponds to inertia-gravity modes, of the system of equations governing a rotating and self-gravitating gas planet. With certain boundary conditions, we rigorously and precisely characterize the essential spectrum and show how it splits from the portion of the spectrum corresponding to the acoustic modes. The fundamental mathematical tools in our analysis are a generalization of the Helmholtz decomposition and the Lopantinskii conditions.
\end{abstract}

%\keywords{Essential spectrum, inertia-gravity modes, Helmholtz decomposition, Lopatinskii conditions}

%\subjclass{35Q85, 47A10, 35G35, 35R09}

\section{Introduction}

We characterize the spectrum of rotating gas planets \tred{with a nonvanishing surface density, which we will call truncated,} and the spectral component associated with fluid (outer) cores of rotating terrestrial planets. \tred{In the case of a polytropic model, truncation yields such a nonvanishing surface density.} As commonly done, we assume the absence of viscosity of the fluid. We focus on determining the essential spectrum associated with gravito-inertial (gi) modes, in addition to the discrete spectrum associated with acoustic (p) modes, starting from the acoustic-gravitational system of linear equations for seismology supplemented with appropriate boundary conditions in a rotating reference frame. That is, in the case of gas planets we impose a vanishing (Lagrangian) pressure boundary condition. We do not impose incompressibility as the fluid does support acoustic modes. (We note that inertia-gravity modes in fluid cores of terrestrial planets are sometimes referred to as undertones \cite{CrossleyRochester:1980}.) While gravity modes owe their existence to a buoyancy force and inertial modes use Coriolis force as the restoring force, the inertia-gravity spectrum is controlled by both the Coriolis and buoyancy forces. For the characterization of the spectra we introduce a modification of the classical Helmholtz decomposition and Leray projector (with range reminiscent of the anelastic approximation) while assuming ``general'' spatial variability in the parameters such as density of mass and Brunt-V\"{a}is\"{a}l\"{a} frequency. Related work to the study of inertia-gravity modes, which we \tred{will} discuss below, has taken the linearized hydrodynamics equations as a point of departure.

Assuming an incompressible fluid and homogeneity, inertial waves were studied in a rotating sphere of fluid in the laboratory \cite{Greenspan:1965, AldridgeToomre:1969}. Kudlick \cite{Kudlick:1966} found an implicit solution for the eigenfrequencies of the inertial modes of a contained fluid spheroid, and Greenspan \cite{Greenspan:1969} calculated a pure point dense spectrum for the Poincar\'{e}'s problem (called so after Cartan \cite{Cartan:1922} who followed Poincar\'{e}'s paper \cite{Poincare:1885}) for cylindrical and spherical configurations. \tred{Bryan \cite{bryan1889vi} was the first to
give the eigenfrequencies of an incompressible axisymmetric ellipsoid including
surface gravity waves. However, in the present work we do not consider the incompressible case.} Ralston \cite{Ralston:1973} studied the spectrum of the generator of the group of motions of an inviscid fluid in a slowly rotating container, and of axisymmetric motions of a large rotating ring of fluid. He presented a family of examples exhibiting various mixtures of continuous and point spectra for this case. Colin de Verdi\`{e}re and Vidal \cite{CdVVidal:2023} reproved the fact, due to Backus and Rieutord \cite{BackusRieutord:2017}, that the Poincar\'{e} operator in ellipsoids admits a pure point spectrum with polynomial eigenfunctions. (Rapidly rotating fluid masses are usually ellipsoidal at the leading order \cite{Chandrasekhar:1969} because of centrifugal forces and, possibly, tidal interactions due to orbital partners.) They then showed that the eigenvalues of this operator restricted to polynomial vector fields of fixed degree admit a limit repartition given by a probability measure that they construct explicitly. In this context, we mention the work by Ivers \cite{Ivers:2017} on the enumeration, orthogonality and completeness of the incompressible Coriolis modes in a tri-axial ellipsoid, and the work by Maffei, Jackson and Livermore \cite{MaffeiJacksonLivermore:2017} on the characterization of columnar inertial modes in rapidly rotating spheres and spheroids. A WKBJ formalism, under the assumption of a spherically symmetric structure, for inertial modes in rotating bodies and a comparison with numerical results were developed by Ivanov and Papaloizou \cite{IvanovPapaloizou:2010}.

We mention the work by Rieutord and Noui \cite{RieutordNoui:1999} studying the analogy between gravity modes and inertial modes in spherical geometry and the work by Dintrans, Rieutord, and Valdettaro \cite{DintransRieutordValdettaro:1999} on gravito-inertial waves in a rotating stratified sphere or spherical shell. Vidal and Colin de Verdi\`{e}re \cite{VidalCdV:2024} studied the inertia-gravity oscillations that can exist within pancake-like geophysical vortices; they considered a fluid enclosed within a triaxial ellipsoid which is stratified in density with a constant Brunt-V\"{a}is\"{a}l\"{a} frequency. \tred{Results of the present work include the case considered in \cite{VidalCdV:2024} if we replace the gravitational field of a planet by an externally imposed gravitational field on a rotating fluid.}

Studies of internal oscillations and inertia-gravity modes specifically pertaining to the Earth's fluid (outer) core, (again) with simple models, with different boundary conditions date back to the work of Olson \cite{Olson:1977} and Friedlander \cite{Friedlander:1985}. (Before that, Friedlander and Siegmann had studied internal oscillations in a contained rotating stratified fluid \cite{FriedlanderSiegmann:1982a} and in  a rotating stratified fluid in an arbitrary gravitational field \cite{FriedlanderSiegmann:1982b}. Seyed-Mahmoud, Moradi, Kamruzzaman and Naseri \cite{Seyed-Mahmoud-etal:2015} studied numerically axisymmetric compressible and stratified fluid core models with different stratification parameters in order to study the effects of the core’s density stratification on the frequencies of some of the inertia-gravity modes of this body.

%\textcolor{red}{Add reference \cite{goodman2009dynamical}}

WKB asymptotics of inertia-gravity modes -- as well as the ray dynamics including attractors where these modes concentrate or exhibit some singularities -- specifically pertaining to axisymmetric stars (but gas planets alike) without rigid boundaries were developed by Prat et al. \cite{Prat-etal:2016,Prat-etal:2018}. Colin de Verdi\`{e}re and Saint-Raymond \cite{CdVSaintRaymond:2020} investigated spectral properties of $0$th order pseudodifferential operators under natural dynamical conditions motivated by the study of (linearized) internal waves on tori. Dyatlov and Zworski \cite{DyatlovZworski:2019} provided proofs of their results based on the analogy to scattering theory. 

Here, we depart from assumptions invoking symmetry or homogeneity, when the component of the spectrum associated with inertia-gravity modes is no longer pure point or everywhere dense, and aim to characterize the essential spectrum in generality, starting from the full acoustic-gravitational system of (linear) equations supplemented with appropriate boundary conditions in a rotating reference frame. Indeed, the main challenge in the analysis is honoring the boundary conditions. We will not rely on any knowledge of expressions for the modes (generalized eigenfunctions), while this knowledge was essential in many of the works referenced above. In pioneering work, Valette \cite{valette1987spectre,Valette:1989a} presented an initial characterization of the spectrum on similar grounds. We build this, providing a mathematical description of the essential spectrum. In particular, we prove that the essential spectrum has no contributions away from the real and imaginary axes, and the portions on the axes are bounded within a certain region. Our method begins with decomposition into a component with zero dynamic pressure and a corresponding potential component. Using the Schur complement as in \cite{tretter2008spectral}, we see that this decomposition naturally splits the spectrum into one portion that can be associated \tred{with} acoustic modes, and a second component comprising intertia-gravity modes. The decomposition allows us to prove that the essential spectrum entirely corresponds with the inertia-gravity component. The inertia-gravity modes are then analyzed first using techniques from microlocal analysis due to Colin de Verdi\`{e}re \cite{CdV:2020} in the interior, and then by reformulation into a large system of PDEs to handle the boundary. We are able to show that, when certain ellipticity conditions are satisfied, this system satisfies the Lopatinskii conditions \cite{agmon1964estimates}, and this allows us to precisely determine the essential spectrum in Theorem~\ref{thm:pp}. We also consider the bounds on the full spectrum of \cite{dyson1979perturbations}, and provide an alternate proof of those estimates, adapted to the situation we consider, in Proposition \ref{prop:DS}. Finally, we provide a partial resolution of the identity in terms of acoustic modes as this plays a role in seismological studies.

It has been noted that eigenfrequencies associated with rotational modes (originating from Liouville's equations), which we will not consider here, are embedded in the essential spectrum \cite{RogisterValette:2009}. It will be interesting to study viscosity limits \cite{GalkowskiZworski:2022}, which we leave for future work. \tred{Literature already exists in this direction, including Rieutord, Valdettaro and Georgeot \cite{rieutord2002analysis} who gave an explicit solution
in the limit of vanishing viscosity for modes around attractors in a
spherical shell.}
% Our results provide full rigorous characterization of the essential spectrum for the complete acoustic-gravitational system of linear equations with vanishing Langrangian pressure boundary conditions. 
We mention that analysis of the essential spectrum for Maxwell's equations with conductivity has also been carried out \cite{lassas1998essential,alberti2019essential} in both bounded and unbounded domains using some of the same tools. 
% In particular, the Helmholtz decomposition, similar to our decomposition, plays a prominent role.

% the phenomenon of kinetic energy localization of internal waves in rotating fluids, and wave attractors, further study 

% quasi-geostrophic modes in the Earth's fluid core with an outer stably stratified layer (Vidal and Schaeffer)

\section{Acousto-gravitational system of equations and well-posedness}

We consider the linearised hydrodynamics arising from perturbations of rotating self-gravitating truncated gas planets. Here the truncation is realised by setting the pressure equal to zero at the surface similar to \cite{dewberry2021constraining}. The displacement vector of a gas or liquid parcel between the
unperturbed and perturbed flow is $\displacement{}$. The unperturbed values of pressure ($p$), density ($\rho$) and gravitational potential ($\Phi$) are denoted with a zero subscript. We have\tred{, in a coordinate system rotating with the planet,}~\footnote{This form follows from the incremental Lagrangian stress formulation in the acoustic limit \tred{as presented, for example, in \cite{de2015system}.}} \cite{Dalsgard2002,dewberry2021constraining}
\begin{equation}
   \rho_0 \partial_t^2 \displacement{}
     + 2 \rho_0 \Omega \times \partial_t \displacement{}
   = \nabla({\kappa \nabla \cdot {\displacement{}}})
   - \nabla (\rho_0 {\displacement{} \cdot \nabla (\Phi_0 + \Psi^s)})
   + (\nabla \cdot (\rho_{0} \displacement{}))
     \nabla (\Phi_0 + \Psi^s)
   - \rho_0 \nabla\Phi' ,
\label{eq: MomentumConservation6}
\end{equation}
where \tred{$\kappa = p_0 \gamma$ is the bulk-modulus and $\gamma$ is the adiabatic index,}
\begin{equation}
   \nabla^2 \Phi' = -4\pi G \nabla \cdot (\rho_{0} \displacement{})
\label{eq: PoissonEqPert-2}
\end{equation}
and $\Psi^s$ denotes the centrifugal potential,
\begin{equation}
   \Psi^s = -\tfrac{1}{2} (\Omega^2 x^2 - (\Omega \cdot x)^2)
   \label{eq: RotPot}
\end{equation}
($|\Omega|$ signifying the rotation rate of the planet). We may
introduce the solution operator, $S$, such that
\begin{equation}
     \Phi' = S(\rho_0 \displacement{}) .
       \label{eq: PerturbGravOperator}
\end{equation}
We will use the shorthand notation,
\begin{equation}
   g_0' = -\nabla (\Phi_0 + \Psi^s).
    \label{eq: PerturbGrav}
\end{equation}
\tred{In fact our results and method of proof also apply in the so-called f-plane approximation\footnote{\tred{That is, a local approximation assuming $g_0' \cdot \Omega$ is constant.}} as considered in \cite{VidalCdV:2024}.} In the above, $\Omega$, $\rho_0$, $\Phi_0$ and $\kappa$ are known
(unperturbed) quantities. We recognize the acoustic wave speed,
\begin{equation}
   c^2 = \kappa \rho_0^{-1} .
   \label{eq: SoundSpeed}
\end{equation}
Typically, the underlying manifold, $M$ say, is a spheroid with the axis of rotation aligned with $\Omega$. A spherically symmetric manifold requires $\Omega = 0$ from well-posedness arguments.

%\subsection*{Brunt-V\"ais\"al\"a frequency}

We rewrite the first two terms on the right-hand side of (\ref{eq:
  MomentumConservation6}),
\begin{equation}
   \nabla({\kappa \nabla \cdot {\displacement{}}})
   - \nabla (\rho_0 {\displacement{} \cdot \nabla (\Phi_0 + \Psi^s)})
   = \nabla [\kappa \rho_0^{-1} \,
             (\nabla \cdot (\rho_0 \displacement{})
     - \tilde{s} \cdot \displacement{})] ,
\end{equation}
in which
\begin{equation}
   \tilde{s} = \nabla \rho_0 - g_0' \frac{(\rho_0)^2}{\kappa};
\end{equation}
$\tilde{s}$ is related to the Brunt-V\"ais\"al\"a frequency
\begin{equation}
   N^2 = \rho_0^{-1} (\tilde{s} \cdot g_0').
\label{eq: N2stilde}
\end{equation}
We may identify $-\kappa \rho_0^{-1} \, (\nabla \cdot (\rho_0
\displacement{}) - \tilde{s} \cdot \displacement{})$ with the dynamic pressure, $P$ say. (The so-called reduced pressure is given by $\rho_0^{-1} P + \Phi'$.) Thus (\ref{eq: MomentumConservation6}) takes the form
\begin{equation}
   \partial_t^2 (\rho_0 \displacement{})
     + 2 \Omega \times \partial_t (\rho_0 \displacement{})
   = \nabla [c^2 \,
             (\nabla \cdot (\rho_0 \displacement{})
     - \tilde{s} \cdot  \displacement{})]
     - (\nabla \cdot (\rho_{0} \displacement{})) g_0'
     - \rho_0 \nabla\Phi' .
\label{eq: MCts}
\end{equation}
At high eigenfrequencies, we can suppress the Coriolis-force term, $2 \Omega \times \partial_t (\rho_0 \displacement{})$ and invoke \textit{the Cowling approximation, when one neglects $\nabla\Phi'$}. However, we retain these terms for the moment and rewrite (\ref{eq: MCts}) as
\begin{equation}
   \partial_t^2 (\rho_0 \displacement{})
     + 2 \Omega \times \partial_t (\rho_0 \displacement{})
   = \nabla [c^2 \,
             (\nabla \cdot (\rho_0 \displacement{})
     - \tilde{s} \cdot  \displacement{})]
     - (\nabla \cdot (\rho_0 \displacement{})) g_0'
     - \rho_0 \nabla S(\rho_0 \displacement{}) .
\label{eq: MCx}
\end{equation}
%In reality, $c^2$ and $\rho_0$ approach 0 close to the boundary $\partial M$. In fact, this property can define $\partial M$ (without boundary conditions) which is identified as the surface of a planetary model. Additionally, 
For typical models of gas giants $\tilde{s}$ and $N^2$ are zero in a finite-thickness \tred{shell} in the outer part of the planet, and we will make this assumption. In polytropic models, it can also be shown that density scales as $D^n$ near the surface of the planet, where $D$ is the depth and $n$ is the polytropic index. Furthermore, $c^2$ tends linearly to zero at the boundary independent of the polytropic index. Though such models, where $c^2$ and $\rho_0$ tend to zero at the boundary, are \tred{more} realistic, for simplicity in this article we will apply a different boundary condition.
%Indeed, we assume $\rho_0$ and $c^2$ do not vanish at the boundary, and 
We assume the free-surface boundary condition given by the vanishing of the Lagrangian pressure perturbation, $\kappa (\nabla \cdot u)$, holds and thus get the boundary condition
\begin{equation} \label{eq:BC}
    [\nabla \cdot u]|_{\partial M} = 0,
\end{equation}
or in terms of the mass-motion,
\begin{equation}\label{eq:Bdmassmotion}
    [\nabla \cdot (\rho_0 \displacement{}) - \displacement{} \cdot \nabla \rho_0]|_{\partial M} = 0 .
\end{equation}
%\textcolor{red}{[NOTE: Colin de Verdiere has $n \cdot u = 0$ at the boundary in ellipse paper.]}
\noindent
% For convenience of notation, we introduce the mass-motion, $x$,
% \begin{equation}
%    x = \rho_0 \displacement{} .
% \label{eq: MassMotion}
% \end{equation}
Analyzing the spectrum follows replacing the operator on the left-hand side of \eqref{eq: MCx} by
\begin{equation} \label{eq:Fldef}
    \lambda^2 \operatorname{Id}
         + 2 \lambda R_{\Omega} =: F(\lambda),\quad
   R_{\Omega} \partial_t (\rho_0 \displacement{})
   = \Omega \times \partial_t (\rho_0 \displacement{}) ,
\end{equation}
that is, upon replacing $\partial_t$ by $\lambda$. We introduce the shorthand notation
\begin{equation} \label{eq:A2x}
    \rho_0 A_2 (\displacement{}) := -\nabla [c^2 \,
             (\nabla \cdot (\rho_0 \displacement{})
     -  \tilde{s} \cdot \displacement{})]
     + (\nabla \cdot (\rho_0 \displacement{})) g_0'
     + \rho_0 \nabla S(\rho_0 \displacement{}) ,
\end{equation}
and
\begin{equation} \label{eq:Llambda}
    L(\lambda) = F(\lambda) + A_2
\end{equation}
identified as a quadratic operator pencil.

Let us now analyze the operator $A_2$ introduced above. We will consider the weak form of $A_2$ on the Hilbert space $H = L^2(\rho_0 \dd x)^3$.
% \[
% H = \left\{ u \in L^2(\rho_0 \dd x)^3\ :\
%             \int_M u \, \rho_0 \dd x = 0 \right\} .
% \]
% Note that $H$ is in the orthogonal complement of the rigid body translations in $L^2(\rho^0 \dd x)$.
For $u$ and $v \in H$ sufficiently regular, using integration by parts gives
\[
\begin{split}
\Big \langle v, A_2(u) \Big \rangle_H & = \int_M \kappa (\nabla \cdot \overline{v}) (\nabla \cdot u ) + \rho_0 (\nabla \cdot \overline{v}) (g_0' \cdot u) + (g_0' \cdot \overline{v})(\nabla \cdot (\rho_0 u)) + \rho_0 \overline{v} \cdot \nabla S(\rho_0 u) \ \mathrm{d} x\\
&\hskip1cm - \int_{\partial M} (n \cdot \overline{v} ) \Big (  \kappa  (\nabla \cdot u) + \rho_0 (g_0'\cdot u) \Big ) \ \mathrm{d} s.
\end{split}
\]
By the proof of \cite[Lemma 4.1, Eqn (4.10)]{de2015system} we can rewrite the gravitational term to obtain
\[
\begin{split}
\Big \langle v, A_2(u) \Big \rangle_H & = \int_M \kappa (\nabla \cdot \overline{v}) (\nabla \cdot u ) + \rho_0 (\nabla \cdot \overline{v}) (g_0' \cdot u) + (g_0' \cdot \overline{v})(\nabla \cdot (\rho_0 u))  \ \mathrm{d} x\\
& \hskip1cm - \frac{1}{4 \pi G} \int_{\mathbb{R}^3}\nabla S(\rho_0\overline{v}) \cdot \nabla S ( \rho_0 u) \ \mathrm{d} x - \int_{\partial M} (n \cdot \overline{v} ) \Big (  \kappa  (\nabla \cdot u) + \rho_0 (g_0'\cdot u) \Big ) \ \mathrm{d} s
\end{split}
\]
If we assume boundary condition \eqref{eq:BC}, we can obtain the quadratic form
\begin{equation}\label{eq:a2}
\begin{split}
a_2(v,u) = \Big \langle v, A_2(u) \Big \rangle_H & = \int_M \kappa (\nabla \cdot \overline{v}) (\nabla \cdot u ) + \rho_0 (\nabla \cdot \overline{v}) (g_0' \cdot u) + (g_0' \cdot \overline{v})(\nabla \cdot (\rho_0 u))  \ \mathrm{d} x\\
& \hskip1cm - \frac{1}{4 \pi G} \int_{\mathbb{R}^3}\nabla S(\rho_0 \overline{v}) \cdot \nabla S (\rho_0 u) \ \mathrm{d} x - \int_{\partial M} (n \cdot \overline{v} ) \rho_0 (g_0'\cdot u) \ \mathrm{d} s
\end{split}
\end{equation}
which is symmetric because $g_0'$ and $\nabla \rho_0$ are parallel in $M$ and on the boundary $\partial M$ the vectors $g_0'$ and $n$ are parallel (see \cite[Lemma 2.1]{de2015system}). Also, $a_2$ can be extended to a bounded sequilinear form on domain $H_{\Div}(M,L^2(\partial M))$,
which is the closure of $C^\infty(M) \cap H$ under the inner-product
\[
\langle v, u \rangle_{H_{\Div}(M,L^2(\partial M))} = \int_M \Big( (\nabla \cdot \overline{v}) (\nabla \cdot u ) + \overline{v} \cdot u \Big ) \ \rho_0\ \mathrm{d}x + \int_{\partial M} (n\cdot \overline{v})(n \cdot u) \ \mathrm{d} s.
\]
Under some hypothesis, it is then true that $a_2$ is $H$-coercive on $H_{\Div}(M,L^2(\partial M))$. The properties mentioned above are summarised in the next lemma (see also \cite{de2015system} for a similar, but more complicated, case).

\begin{lemma}\label{lem:a2}
    Suppose that $M$ is compact with smooth boundary $\partial M$, $c$, $g_0' \in C(M)$, $\rho_0 \in C^1(M)$,  $c^2> 0$ on $M$, $g_0'$ and $\nabla \rho_0$ are parallel in $M$, $g_0'$ and $n$
are parallel on $\partial M$, and $g_0'\cdot n < 0$ on $\partial M$. Then $a_2$ defined by \eqref{eq:a2} is a continuous sequilinear form on $H_{\Div}(M,L^2(\partial M))$. Furthermore, there exist constants $\alpha$, $\beta > 0$ such that
\begin{equation} \label{eq:coerc}
    a_2(u,u) \geq \alpha \|u\|_{H_{\Div}(M,L^2(\partial M))}^2 - \beta \|u\|_H^2
\end{equation}
for all $u \in H_{\Div}(M,L^2(\partial M))$.
\end{lemma}

\begin{remark}
    The hypothesis that $c^2>0$ is not realistic for gas giants as noted above since in that case $c^2$ will go to zero at $\partial M$. The requirements that the given vectors are parallel follow from hydrostatic equilibrium \cite[Lemma 2.1]{de2015system}. \tred{Note that hydrostatic equilibrium also implies barotropic equilibrium; that is, $g_0'$, $\nabla \rho_0$ and $\nabla p_0$ are all parallel meaning level surfaces of $\rho_0$ and $p_0$ will coincide.} 
\end{remark}

\begin{proof}
    Sesquilinearity follows from the hypotheses that certain vectors are parallel as can be seen from \eqref{eq:a2}. Also, continuity is proven by directly applying the Cauchy-Schwartz inequality to \eqref{eq:a2}. Let us now establish \eqref{eq:coerc}.
    
    In the rest of the proof, $C$ and $D$ will always be positive constants which may change from step to step. Since $M$ is compact, $c^2 \geq C > 0$ and so
    \[
    \int_M \kappa (\nabla \cdot u) (\nabla \cdot u ) \ \mathrm{d} x \geq C \| \nabla \cdot u \|_H^2.
    \]
    Applying the Cauchy-Schwartz inequality and using bounds on $|g_0'|$, $\rho_0$ and $|\nabla \rho_0|$ gives, for any $\epsilon >0$
    \[
    \left | \int_M \rho_0 (\nabla \cdot u) (g_0' \cdot u) + (g_0' \cdot u)(\nabla \cdot (\rho_0 u))  \ \mathrm{d} x \right | \leq C \epsilon \|\nabla \cdot u \|_H^2 + D  \epsilon^{-1} \|u\|_H^2.
    \]
    Applying the definition \eqref{eq: PerturbGravOperator} of $S$, we can bound the gravitational term
    \[
    \left | \frac{1}{4 \pi G} \int_{\mathbb{R}^3}\nabla S(\rho_0 u) \cdot \nabla S (\rho_0 u) \ \mathrm{d} x \right | \leq C \| u \|_H^2,
    \]
    and the fact that $g_0'$ and $n$ are parallel as well as hypothesis $g_0'\cdot n < 0$ implies
    \[
    - \int_{\partial M} (n \cdot u ) \rho_0 (g_0'\cdot u) \ \mathrm{d} s \geq C \|u\|^2_{L^2(\partial M)}.
    \]
    Combining all of the previous estimates and taking $\epsilon$ sufficiently small proves \eqref{eq:coerc}.
\end{proof}

Lemma \ref{lem:a2} shows that $(H_{\Div}(M,L^2(\partial M)),H,a_2)$ is a Hilbert triple, which implies many results about the operator $A_2$ \cite[Chapter VI.3.2.5]{dautray1999mathematical} some of which we collect in the next Corollary.

\begin{corollary} \label{cor:Hilbert3}
    The operator $A_2$ is continuous from $H_{\Div}(M,L^2(\partial M))$ to the Hilbert dual $H_{\Div}(M,L^2(\partial M))'$. Also, $A_2$ is an unbounded self-adjoint operator on $H$ with domain
    \[
    D(A_2) = \{ u \in H_{\Div}(M,L^2(\partial M)) \ : \ v \mapsto a_2(u,v) \mbox{ is continuous with respect to the $L^2(M,\rho_0 \ \mathrm{d} x)$ norm}\}.
    \]
\end{corollary}

By analysing \eqref{eq:a2} and the equation before, we can say more about $D(A_2)$ if we make additional regularity assumptions about the parameters. This is done in the next corollary.

\begin{corollary} \label{cor:DA2}
    In addition to the hypotheses of Lemma \ref{lem:a2}, assume that $c^2 \in C^1(M)$ and $\rho_0 \in C^2(M)$. Then
    \[
    D(A_2) = \{ u \in H_{\Div}(M,L^2(\partial M)) \: | \: \nabla [c^2 \,
             (\nabla \cdot (\rho_0 \displacement{})
     - \tilde{s} \cdot  \displacement{})] \in L^2(M), \ [\nabla 
    \cdot u]|_{\partial M} = 0 \}.
    \]
\end{corollary}

\section{Decompositions of Hilbert space and spectrum}

%As observed in the previous section above Corollary \ref{cor:Hilbert3} we have the setting of a Hilbert triple,
%$$
%  H_{\Div}(M,L^2(\partial M)) \hookrightarrow H
%     \hookrightarrow H_{\Div}(M,L^2(\partial M))'.
%$$
% where
% \[
%    H = \left\{ u \in L^2(\rho^0 \dd x)\ :\
%             \int u \, \rho^0 \dd x = 0 \right\} ,
% \]
% \[
% H_{\Div} = \left\{ u \in L^2(\rho^0 \dd x)\ :\ \nabla \cdot u \in L^2(\rho^0 \ \dd x), \
%             \int u \, \rho^0 \dd x = 0 \right\}
% \]
% and $H_{\Div}'$ is the dual space of $H_{\Div}'$.
%Each space is separable and continuously, densely and injectively embedded in the next.
%%Old stuff
%Here, $E'$ is the Banach dual of $E$. We note that $E$ is separable. We note that $E = H_{\Div}$ and let
%\begin{equation} \label{eq:Ha}
   %H = \left\{ u \in %L^2(\rho^0 \dd x)\ :\
   %         \int u \, \rho^0 \dd x = 0 \right\} ,
%\end{equation}
%which is the orthogonal complement of the rigid body translations in $L^2(\rho^0 \dd x)$.

Our main goal is to characterise the spectrum of the operator pencil $L(\lambda)$ given by \eqref{eq:Llambda}. To this end, let us recall the definition of the resolvent set, spectrum and essential spectrum as given, for example, in \cite{moller2015spectral}.

\begin{definition}\label{def:spect}
Let $\Omega \subset \mathbb{C}$ be an open set and for each $\lambda \in \Omega$ suppose $T(\lambda)$ is a closed linear operator from $D(T(\lambda)) \subset H$ to $H$. The set of $\lambda \in \Omega$ such that $T(\lambda)$ is bijective on its domain with bounded inverse $T(\lambda)^{-1}:H\rightarrow H$ is the resolvent set of $T$ which we will notate as $\rho(T)$. The complement of the resolvent set is the spectrum $\sigma(T)$, and the set of $\lambda \in \Omega$ such that $T(\lambda)$ is not a Fredholm operator\tred{\footnote{\tred{A closed linear operator between Hilbert spaces is called Fredholm if its nullspace is finite dimensional and its range has finite codimension. Thus, at points in the essential spectrum there may be infinitely many linearly independent eigenfunctions.}}} is the essential spectrum $\sigma_{ess}(T)$.
\end{definition}

\noindent
Note that by Lemma \ref{lem:a2} and the Lax-Milgram Theorem (see for example \cite[Theorem 7, p 368]{dautray1999mathematical}), for $\mathrm{Re}\ \lambda^2 > \beta$ we know that $L(\lambda)$ has bounded inverse and so such $\lambda$ are in the resolvent set $\rho(L)$. However, as we will see below after considering an appropriate decomposition of $H$, $L(\lambda)^{-1}$ is not compact meaning that the analytic Fredholm theory cannot be applied as in, for example, \cite[Lemma 1.2.1, p 7]{moller2015spectral} and the essential spectrum $\sigma_{ess}(L)$ is not empty. 

%\textcolor{blue}{The essential spectrum of the Earth is due to fluid motion of the outer core, and has been studied from the point of view of normal modes (e.g. see \cite{triana2022core}).}

We now develop an orthogonal decomposition generalizing the Helmholtz decomposition,
\begin{equation}\label{eq:Hdecomp}
   H =  H_1 \oplus H_2 ,
\end{equation}
with corresponding projections,
\[
   \pi_{1,2} :\ H \to H_{1,2} ,
\]
with the goal to extract the part of the point spectrum associated
with the acoustic normal modes and characterize the essential spectrum. The construction of the decomposition will entail the introduction of a space $E_1$ such that the injection of $E_1$ into $H_1$ is compact.

We introduce the operator,
\begin{equation} \label{eq:Tdef}
    T u :=  \nabla \cdot (\rho_0 u) - \tilde{s} \cdot u
       = \rho_0 [\nabla \cdot u + \rho_0 \kappa^{-1} g_0' \cdot u] ,\quad
     D(T) = H_{\Div,0}(M) = \{ u \in H_{\Div}(M) \ : \ u \cdot n|_{\partial M} = 0\} ,
\end{equation}
which sets the dynamic pressure to zero (and induces the so-called \textit{anelastic} approximation). The adjoint, $T^*$, of $T$ is given by
\begin{equation} \label{eq:T*}
   T^* \varphi = -(\nabla (\rho_0 \varphi)
                        + \tilde{s} \varphi)
\end{equation}
with $D(T^*) = H^1(M)$. Operators $T$, $T^*$, have the following properties.

\begin{lemma} \label{lem:T-compact}
Assume the hypotheses of Lemma \ref{lem:a2}. Then $\operatorname{Ran}(T^*)$ is closed in $H$ and $H = \operatorname{Ran}(T^*)\oplus \operatorname{Ker}(T)$. Moreover, the map $\Pi: H \rightarrow H^1(M)$ taking $u$ to the unique minimum norm $\varphi \in H^1(M)$ satisfying $u = T^* \varphi + u_2$ for $u_2 \in \operatorname{Ker}(T)$ is continuous. Finally, the injection $\operatorname{Ran}(T^*) \, \cap \, D(T) \hookrightarrow H$, with the $H_{\Div}(M)$ topology on the domain, is compact.
\end{lemma}

\begin{proof}
The decomposition $H =\overline{\operatorname{Ran}(T^*)} \oplus \operatorname{Ker}(T)$ is a general fact for a closed operator $T$ with dense domain. In this case,
$D(T) = H_{\Div,0}(M)$ is dense in $H$ as in \cite[Prop. 2, p. 67]{valette1987spectre} it is stated that $C_0^\infty \subset H_{\Div,0}(M)$ is dense in $H_{\Div}(M)$ and then $\overline{H_{\Div}(M)}=H$. The next part of the proof follows the method of \cite[Chap]{Taylor1}. To show that $\operatorname{Ran}(T^*)$ is closed, consider the operator $\mathcal{L}_{T^*}: H^1(M) \rightarrow H^1(M)^*$ defined by
\[
\langle \mathcal{L}_{T^*} \varphi,\psi \rangle = \langle T^* \varphi, T^* \psi \rangle_H \quad \forall \varphi, \ \psi \in H^1(M).
\]
It is possible to find a real constant $C$ sufficiently large so that
\[
\langle (\mathcal{L}_{T^*}+C) \varphi,\varphi \rangle = \langle T^* \varphi, T^* \varphi \rangle_H + C \langle \varphi,\varphi \rangle_{L^2(\rho_0 \ \mathrm{d} x)} \geq \tilde{C} \|\varphi\|_{H^1(M)}^2 \quad \forall \varphi \in H^1(M).
\]
Following the method of \cite[]{Taylor1}, we deduce that $\mathcal{L}_{T^*} + C$ is one-to-one and onto, and by considering the inverse of $\mathcal{L}_{T^*} + C$, which is a compact and self-adjoint operator, we find that there exists an orthonormal basis of eigenvectors on $H$ for $\mathcal{L}_{T^*}$ with discrete eigenvalues going to infinity. Since $\mathcal{L}_{T^*}$ is a non-negative operator by its definition, the corresponding eigenvalues must all be non-negative with the possibility that zero is an eigenvalue with finite multiplicity. Thus the kernel of $\mathcal{L}_{T^*}$ is finite dimensional, and note that it also coincides with kernel of $T^*$.

Now, suppose that $\{ y_n \}_{n=1}^\infty \subset \operatorname{Ran}(T^*)$ is a Cauchy sequence for $H$. Then there exists a sequence $\{\varphi_n \}_{n=1}^\infty \subset H^1(M)$ such that $y_n = T^* \varphi_n$ and without loss of generality we can assume that all $\varphi_n$ are orthogonal to the kernel of $\mathcal{L}_{T^*}$. By taking the smallest positive eigenvalue of $\mathcal{L}_{T^*}$, which is strictly greater than zero by the above argument, we therefore show that $\varphi_n$ is a Cauchy sequence in $L^2(\rho_0 \ \mathrm{d} x)$. From the formula \eqref{eq:T*} of $T^*$, as well as facts that $\rho_0 \in C^1(M)$ is positive and $\tilde{s} \in C(M)^3$, we have the inequality
\[
\|y_n - y_m\|^2_{H} + C \|\varphi_n - \varphi_m\|^2_{L^2(\rho_0 \ \mathrm{d} x)} \geq C \| \nabla (\varphi_n-\varphi_m) \|^2_{H} 
\]
for some $C>0$ possibly different from the one above. Therefore $\{\varphi_n\}_{n=1}^\infty$ is a Cauchy sequence in $H^1(M)$ and since $H^1(M)$ is complete it must converge. This implies that $\{y_n\}_{n=1}^\infty$ must also converge in $H$ to a point in $\operatorname{Ran}(T^*)$. Since $H$ is a complete space, this proves that $\operatorname{Ran}(T^*)$ is closed.

Using its eigendecomposition, we can find a pseudoinverse for $\mathcal{L}_{T^*}$ which we will denote $\mathcal{L}_{T^*}^{\dagger}: H^1(M)^* \rightarrow H^1(M)$. To prove the continuity of the map $\Pi$, note that, after extending $T$ to an operator from $H$ to $H^1(M)^*$ by duality,
\[
\Pi = \mathcal{L}_{T^*}^\dagger  T.
\]
That the injection, $\operatorname{Ran}(T^*) \, \cap \, D(T) \hookrightarrow H$, is compact follows exactly as in \cite[Prop 3d, p. 70]{valette1987spectre} where it was done for the case $g_0'=0 $ using the  Rellich-Kondrashov compactness theorem for Sobolev spaces. The proof is straightforwardly adapted to $g_0'\neq 0.$
\end{proof}

Given Lemma \ref{lem:T-compact}, we obtain the decomposition \eqref{eq:Hdecomp} by setting $H_1 = \operatorname{Ran}(T^*)$ and $H_2 = \operatorname{Ker}(T)$. This is a generalization of the Helmholtz decomposition which is required for our analysis. Using the proof of Lemma \ref{lem:T-compact}, the projection onto $\operatorname{Ran}(T^*)$, generalizing the notion of irrotational, is given by $T^* \mathcal{L}_{T^*}^\dagger T$ and the projection onto $\operatorname{Ker}(T)$, generalizing the notion of anelastic, is given by $I - T^* \mathcal{L}_{T^*}^\dagger T$. These formulae lead to the next lemma.

% {\color{blue}  We have 
% \begin{align*}T T^*\vp=&-  [  (\rho_0)^2 \Delta  \varphi+3\rho_0(\nabla\rho_0)\cdot \nabla \varphi
%                         +\left( |\nabla \rho_0|^2+\rho_0  \Delta\rho_0+ \rho_0    \nabla \cdot \tilde{s}    -|\tilde{s}|^2 \right)\vp]\end{align*}}

\begin{lemma}\label{lem:pi12} The orthogonal projection operators $\pi_1: H \rightarrow \mathrm{Ran}(T^*) = H_1$ and $\pi_2: H \rightarrow \mathrm{Ker}(T) = H_2$ are zero order pseudodifferential operators in the interior of $M$ with principal symbols given respectively by
\begin{equation} \label{eq:pisymbol}
\sigma_p(\pi_1) = \frac{\xi \xi^T}{|\xi|^2}, \quad \sigma_p(\pi_2) = \mathrm{Id} - \frac{\xi \xi^T}{|\xi|^2}.
\end{equation}
% {\color{blue} If $P_\xi=\sigma(\pi_1)$ then in 
% ``$\pi_1: L^2(\rho_0\dd x) \rightarrow \mathrm{Ran}(T^*)$ and $\pi_2: L^2(\rho_0\dd x) \rightarrow \mathrm{Ker}(T)$''
% $\pi_1$ $\pi_2$ should be swapped. But, it is better to associate $\pi_2$ with $H_2$, as it was in previous version. It is also meant in Theorem 2. Please correct also in the proof!}

% {\color{red} [SEAN] I think we want $H_2 = Ker(T)$ to agree with the definition just after this lemma and then to be consistent with the decomposition of the operator. I will try to make everything consistent with that. It might be good to also swap the ordering in Lemma 1.}
\end{lemma}

\begin{proof}
For $u \in H$, suppose that $u_1 = T^* \Pi(u) = \pi_1(u)$ and $u_2 = \pi_2(u)$. Thus
\[
u = T^* \Pi(u) + u_2 \Rightarrow Tu = TT^* \Pi(u).
\]
$TT^*$ is an elliptic second order differential operator and as such has a pseudodifferential parametrix on the interior of $M$, which is an order $-2$ pseudodifferential operator $(TT^*)^{-1}$ there such that
\begin{equation}\label{eq:TT*inv}
(TT^*)^{-1} T u = \Pi(u) +  Ku
\end{equation}
where $K$ is a smoothing operator in the interior of $M$. Therefore
\[
\pi_1(u) = T^* (TT^*)^{-1} T u - T^* Ku.
\]
This proves that $\pi_1$ is a zero order pseudodifferential operator in the interior of $M$. By looking at the prinicpal symbols of $T$ and $T^*$ and using the composition calculus we conclude that $\sigma_p(\pi_1)$ is as given in \eqref{eq:pisymbol}. Since $\pi_2 = \mathrm{Id} - \pi_1$, the conclusion for $\sigma_p(\pi_2)$ follows as well.
\end{proof}

Our next task is to decompose the operator $A_2$. Towards this end, we introduce
% \begin{equation*}
%    H_2 = \{ u \in D(a_2)\ :\ T u = 0 ,\
%                          \int u \rho^0 \dd x = 0 \} .
% \end{equation*}
\[
E_1 = D(A_2) \cap \operatorname{Ran}(T^*)
\]
%removing rigid body translations in the fluid regions, 
and
\[
E_2 = D(A_2) \cap \operatorname{Ker}(T),
\]
whence\rcolor{
\[
   E_1 \oplus E_2\subset D(A_2).
\]
For the opposite inclusion, consider $u \in D(A_2)$. Since $\pi_2(u) \in \operatorname{Ker}(T)$, we see that $\nabla c^2Tu = 0$ and, since $n$ is parallel with $g_0'$ on $\partial M$,
\[
0 = \rho_0 \kappa^{-1} g_0' \cdot u |_{\partial M} = \nabla \cdot u|_{\partial M}.
\]
Therefore, referring to Corollary \ref{cor:DA2}, $\pi_2(u) \in D(A_2)$ and so
\[
E_1 \oplus E_2= D(A_2).
\]
We} now aim to introduce a corresponding block decomposition of the operator $L(\lambda)$ introduced in \eqref{eq:Llambda}. Indeed, let us define the component operators by
\[
L_{ij}(\lambda) = \pi_i L(\lambda) \pi^*_j
\]
for $i$, $j = 1$ and $2$. Considering $D(A_2)$ in Corollary \ref{cor:DA2} and noting that $\operatorname{Ker}(T) \subset D(A_2)$, we see that
\[
D(L_{i1}(\lambda)) = E_1, \quad D(L_{i2}(\lambda)) = \operatorname{Ker}(T)
\]
for $i = 1$ and $2$. With these operators, we see that $L(\lambda)$ is related to these component operators by
\[
L(\lambda) = 
\begin{pmatrix}
\pi_1^* & \pi_2^*
\end{pmatrix}
\begin{pmatrix} L_{11}(\lambda) & L_{12}(\lambda) \\
L_{21}(\lambda) & L_{22}(\lambda)
\end{pmatrix}
\begin{pmatrix}
\pi_1 \\ \pi_2
\end{pmatrix}
\]
and thus the resolvent set, spectrum and essential spectrum of $L(\lambda)$ is equivalent to the same for the block matrix on the right side of the equation which we label as
\[
\mathcal{L}(\lambda) = \begin{pmatrix} L_{11}(\lambda) & L_{12}(\lambda) \\
L_{21}(\lambda) & L_{22}(\lambda)
\end{pmatrix}, \quad D(\mathcal{L}(\lambda)) = E_1 \oplus \operatorname{Ker}(T).
\]
In the next Proposition we summarise some of the properties of the component operators.
\begin{proposition} \label{prop:Lij}
Suppose that $g_0' \in C(M)$ and $\rho_0 \in C^1(M)$. Then the operators $L_{12}(\lambda): \operatorname{Ker}(T) \rightarrow \operatorname{Ran}(T^*)$ and $L_{22}(\lambda): \operatorname{Ker}(T) \rightarrow \operatorname{Ker}(T)$ are bounded. The operator $L_{21}(\lambda)$ with domain $E_1$ is closable with closure a bounded operator from $\operatorname{Ran}(T^*)$ to $\operatorname{Ker}(T)$. Finally, $L_{11}(\lambda)$ with domain $E_1$ is a Fredholm operator with index $0$ and discrete spectrum consisting of eigenvalues which have finite multiplicity. Further, $L_{11}(\lambda)^{-1}$ is compact on the resolvent set of $L_{11}$.
\end{proposition}

\begin{proof}
Suppose that $u \in \operatorname{Ker}(T)$. Then from \eqref{eq:A2x}, \eqref{eq:Llambda} and \eqref{eq:Tdef},
\begin{equation} \label{eq:LKerT1}
L(\lambda) u = F(\lambda) u + \frac{\tilde{s} \cdot u}{\rho_0} g_0' + \nabla S(\rho_0 u).
\end{equation}
The first two terms on the right side above are clearly bounded as they are only multiplication by bounded quantities. The third term, corresponding to self-gravitation is also bounded by Lemma \ref{lem:-4-} which is proven below. Because the projectors $\pi_1$ and $\pi_2$ are both continuous this proves the boundedness of $L_{12}(\lambda)$ and $L_{22}(\lambda)$ as stated.

Now let us consider $L_{21}(\lambda)$. Taking $u \in E_1$ and $v \in \operatorname{Ker}(T)$ we have
\[
\begin{split}
\langle A_2 u, v \rangle_H & = \langle u, A_2 v \rangle _H \\
& = \left \langle u,  \frac{\tilde{s} \cdot v}{\rho_0} g_0' + \nabla S(\rho_0 v) \right \rangle_H\\
& = \left \langle  \frac{g_0' \cdot u}{\rho_0} \tilde{s} + \nabla S(\rho_0 u), v 
\right \rangle_H.
\end{split}
\]
Since this is true for any $v \in \operatorname{Ker}(T)$, we conclude that
\[
L_{21}(\lambda) u =  \pi_2 \left (F(\lambda) u + \frac{g_0' \cdot u}{\rho_0} \tilde{s} + \nabla S(\rho_0 u)\right ).
\]
Similar to above, this is a bounded operator and so $L_{21}(\lambda)$ extends to a bounded operator from $\operatorname{Ran}(T^*)$ to $\operatorname{Ker}(T)$ as claimed.

Finally, the statement about $L_{11}(\lambda)$ follows by the argument of \cite[Lemma 1.1.11]{moller2015spectral} and the compactness of resolvent of $\pi_1 A_2 \pi_1^*$ which is a consequece of Lemma \ref{lem:T-compact}. 
Indeed, since $E_1 \subset H_{\Div}(M,L^2(\partial M))$ 
by Lemma \ref{lem:a2} $\pi_1 A_2 \pi_1^* + \beta I$ is invertible from its domain $E_1$ into $\operatorname{Ran}(T^*)$ and, by the closed graph theorem,
the inverse is bounded with the graph norm on $E_1$. Since the injection $E_1 \hookrightarrow \operatorname{Ran}(T^*) \cap D(T)$, with the graph norm on $E_1$ and
$H_{\Div}(M)$ topology on $\operatorname{Ran}(T^*) \cap D(T)$ is continuous, by Lemma \ref{lem:T-compact} $(\pi_1 A_2 \phi_1^* + \beta I)^{-1} : \operatorname{Ran}(T^*) \rightarrow \operatorname{Ran}(T^*)$ is compact. Furthermore, we have the identity
\[
\pi_1 A_2 \pi_1^* (\pi_1 A_2 \pi_1^* + \beta I)^{-1} = I - \beta (\pi_1 A_2 \pi_1^* + \beta I)^{-1}
\]
which when applied to $L_{11}(\lambda)$ gives
\begin{equation} \label{eq:L11inv}
L_{11}(\lambda) (\pi_1 A_2 \pi_1^* + \beta I)^{-1} =  I + (\pi_1 F(\lambda) \pi_1^* - \beta I) (\pi_1 A_2 \pi_1^* + \beta I)^{-1}.
\end{equation}
This operator is therefore a compact pertubation of the identity and so $L_{11}(\lambda)$ is Fredholm with index $0$ as claimed in the statement of the proposition. The fact that $L_{11}(\lambda)$ has discrete spectrum consisting of eigenvalues with finite multiplicity then follows from the analytic Fredholm theory. For $\lambda$ in the resolvent set of $L_{11}$, we can apply $L_{11}^{-1}$ to \eqref{eq:L11inv} and get
\[
L_{11}^{-1}(\lambda) = ((1-\beta)I - \pi_1 F(\lambda) \pi_1^*) (\pi_1 A_2 \pi_1^* + \beta I)^{-1}
\]
which is compact.
\end{proof}

\begin{remark} \label{rem:Lij}
If we additionally assume that $g_0'$ and $\nabla \rho_0$ are parallel, which is a requirement for well-posedness of the system \eqref{eq: MomentumConservation6}, and use the Brunt-V\"ais\"al\"a frequency $N^2$ (see \eqref{eq: N2stilde}), the proof of Proposition \ref{prop:Lij} implies the following formulae
\begin{equation}
\begin{split}
L_{12}(\lambda) & = \pi_1 \left ( F(\lambda) + N^2 \hat{g}_0' \hat{g}_0'^T + \nabla S \rho_0 \right ) \pi_2^*, \quad L_{22}(\lambda) = \pi_2 \left ( F(\lambda) + N^2 \hat{g}_0' \hat{g}_0'^T + \nabla S \rho_0 \right ) \pi_2^*,\\
L_{21}(\lambda) &= \pi_2 \left ( F(\lambda) + N^2 \hat{g}_0' \hat{g}_0'^T + \nabla S \rho_0 \right ) \pi_1^*,
\end{split}
\end{equation}
where
\[
\hat{g}_0' = \frac{g_0'}{\|g_0'\|}.
\]
From these formulae and Lemma \ref{lem:-4-}, $L_{22}(\lambda)$ cannot have a compact inverse. Thus by taking $u \in \operatorname{Ker}(T)$ we see that $L(\lambda)^{-1}$ cannot be compact as observed earlier.
\end{remark}

Before continuing we proof a technical lemma which was used in the proof of Proposition \ref{prop:Lij}, and will also be important below.

\begin{lemma} \label{lem:-4-}
The map $u \to S(\rho_0 u)$ is continuous from $L^2(M)$ to $L^6(\RR^3)$ and the map $u \to \nabla S(\rho_0 u)$ is continuous from
$L^2(M)$ to $L^2(\RR^3)$. The map $u \to \nabla S(\rho_0 u)$ is
compact from $H_{\Div}(M,L^2(\partial M))$ to $L^2(\RR^3)$.
\end{lemma}

\begin{proof} Starting from the definition given by \eqref{eq: PoissonEqPert-2} and \eqref{eq: PerturbGravOperator} of $S$, setting $v = \rho_0 u$, and using Parseval's identity we have
\begin{equation} \label{eq:DSest}
    \begin{split}
        \|\nabla S(v) \|^2_{L^2(\mathbb{R}^3)^3} = (4 \pi G)^2 \int_{\mathbb{R}^3} \frac{|\hat{v}(\xi) \cdot \xi|^2}{|\xi|^{2}} \dd \xi \leq (4 \pi G)^2 \int_{\mathbb{R}^3} |\hat{v}(\xi)|^2 \dd \xi = (4 \pi G)^2 \|v \|_{L^2(M)^3}^2.
    \end{split}
\end{equation}
From \eqref{eq:DSest}, we conclude that
\begin{enumerate}[label=(\roman*)]
\item
$u \mapsto \nabla S(\rho_0 u)$ is a bounded operator from $L^2(M)^3$ to $L^2(\mathbb{R}^3)^3$, and accordingly that \\[-0.2cm]
\item
$u \mapsto S(\rho_0 u)$ is a bounded operator from $L^2(M)^3$ into $L^6(\mathbb{R}^3)^3$ using the Gagliardo–Nirenberg–Sobolev inequality.
\end{enumerate}
Thus the first two assertions of the lemma have been proven.

Regarding the third assertion, we use \eqref{eq: PoissonEqPert-2} together with Parseval's identity, or integration by parts, in writing
\begin{equation} \label{eq:DSid}
\|\nabla S(v) \|^2_{L^2(\mathbb{R}^3)^3} = 4 \pi G \langle \nabla S(v), v \rangle_{L^2(\mathbb{R})^3} = 4 \pi G \langle \nabla S(v), v \rangle_{L^2(M)^3}.
\end{equation}
For $v\in H_{\Div}(M,L^2(\partial M))$, integration by parts gives
\begin{equation}\label{eq:dsintbypart}
\langle \nabla S(v), v \rangle_{L^2(M)^3} = \langle S(v), n\cdot v \rangle_{L^2(\partial M)} - \langle S(v), \nabla\cdot v \rangle_{L^2(M)^3}.
\end{equation}
Now consider any bounded sequence $\{ u^\ell\} \subset H_{\Div}(M,L^2(\partial M))$. By Alaoglu's Theorem, there is a subsequence that converges weakly in $H_{\Div}(M,L^2(\partial M))$ to a point $\widetilde{v} \in H_{\Div}(M,L^2(\partial M))$. Then the subsequence $\{v^{\ell_k} - \widetilde{v} \}$ is bounded in $H_{\Div}(M,L^2(\partial M))$ and converges weakly to zero. Now, note that by (i) and (ii), and using also H\"older's inequality, $S:L^2(M)^3 \rightarrow H^1(M)^3$ is continuous and so $S:L^2(M)^3 \rightarrow L^2(M)^3$ is compact as well as the composition of $S$ with restriction to the boundary. Thus, by taking a further subsequence if necessary, $\{S(v^{\ell_k} - \widetilde{v})\}$ converges strongly to some point $w \in L^2(M)^3$. Combining \eqref{eq:DSid} and \eqref{eq:dsintbypart} and applying to $v^{\ell_k} - \widetilde{v}$, we have
\[
\|\nabla S (v^{\ell_k}) - \nabla S (\widetilde{v})\|^2_{L^2(\mathbb{R}^3)^3} = \langle S(v^{\ell_k} - \widetilde{v}), n\cdot (v^{\ell_k} - \widetilde{v})\rangle_{L^2(\partial M)} - \langle S(v^{\ell_k} - \widetilde{v}), \nabla\cdot(v^{\ell_k} - \widetilde{v})\rangle_{L^2(M)^3}.
\]
Since $S(v^{\ell_k} - \widetilde{v})$ converges strongly to $w$, $v^{\ell_k} - \widetilde{v}$ is bounded and $v^{\ell_k} - \widetilde{v}$ converges weakly to zero, we obtain that
\[
\lim_{k\rightarrow \infty} \|\nabla S (v^{\ell_k}) - \nabla S (\widetilde{v})\|_{L^2(\mathbb{R}^3)^3} = 0.
\]
This completes the proof.

\end{proof}

% The injection, $E_1 \hookrightarrow H_1$, is compact. [TO DO: For decomposition to be well-defined, we should have $\pi_i : H_i \rightarrow D(A_2)$.] This decomposition implies a decomposition of $L(\lambda)$,
% \[
%    \begin{pmatrix} L_{11}(\lambda) & L_{12}(\lambda) \\
%                  L_{21}(\lambda) & L_{22}(\lambda) \end{pmatrix} ,
% \]
% following the decompositions
% \[
%    \begin{pmatrix} R_{\Omega;11} & R_{\Omega;12} \\
%                 R_{\Omega;21} & R_{\Omega;22} \end{pmatrix} ,\quad
%    \begin{pmatrix} A_{2;11} & A_{2;12} \\
%                 A_{2;21} & A_{2;22} \end{pmatrix} .
% \]
% Consistently with the fact that $A_2$ is self adjoint and that
% $R_{\Omega}$ is antisymmetric, we must have
% \begin{equation} \label{eq:AIJsym}
%    A_{2;11}^* = A_{2;11} ,\
%    A_{2;22}^* = A_{2;22} ,\
%    A_{2;21}^* = A_{2;12}
% \end{equation}
% and
% \begin{equation} \label{eq:RIJsym}
%    R_{\Omega;11}^* = -R_{\Omega;11} ,\
%    R_{\Omega;22}^* = -R_{\Omega;22} ,\
%    R_{\Omega;21}^* = -R_{\Omega;12} .
% \end{equation}

%{\color{red} @Sean: please polish this paragraph which should provide a proof of the second statement in the next Proposition} {\color{yellow} Added the following}

We now apply Frobenius-Schur factorization to the operator $\mathcal{L}(\lambda)$ to draw conclusions about the decomposition of its spectrum. Our factorization and resulting spectral decomposition are essentially the same as \cite[Theorem 2.2.14]{tretter2008spectral} except we consider quadratic dependence on $\lambda$. This does not introduce any serious complication into the method. Suppose that $\rho_1$ is the resolvent set of $L_{11}$ and $\rho_2$ the resolvent set of $L_{22}$ with complements $\sigma_1$ and $\sigma_2$ the corresponding spectra. For $\lambda \in \rho_1$ we define the Schur complement
\[
S_2(\lambda) = L_{22}(\lambda) - L_{21}(\lambda) L_{11}(\lambda)^{-1} L_{12}(\lambda)
\]
and similarly for $\lambda \in \rho_2$ we have
\[
S_{1}(\lambda) = L_{11}(\lambda) - L_{12}(\lambda) L_{22}(\lambda)^{-1} L_{21}(\lambda).
\]
Then for $\lambda \in \rho_1$
\[
\mathcal{L}(\lambda) = \begin{pmatrix}
I & 0 \\
L_{21}(\lambda) L_{11}^{-1}(\lambda) & I
\end{pmatrix}
\begin{pmatrix}
L_{11}(\lambda) & 0 \\
0 & S_2(\lambda)
\end{pmatrix}
\begin{pmatrix}
I & L_{11}^{-1}(\lambda) L_{12}(\lambda) \\
0 & I
\end{pmatrix},
\]
while for $\lambda \in \rho_2$
\begin{equation}\label{eq:Schur2}
\mathcal{L}(\lambda) = \begin{pmatrix}
I & L_{12}(\lambda) L_{22}^{-1}(\lambda) \\
0 & I
\end{pmatrix}
\begin{pmatrix}
S_1(\lambda) & 0 \\
0 & L_{22}(\lambda)
\end{pmatrix}
\begin{pmatrix}
I & 0 \\
L_{22}^{-1}(\lambda) L_{21}(\lambda) & I
\end{pmatrix}.
\end{equation}
By \cite[Lemma 2.3.2]{tretter2008spectral}, since the matrix operators on the outside of the products in each equality above are invertible on the respective resolvent sets, we obtain
\[
\sigma(L) \setminus \sigma_1 = \sigma(S_2),
\quad
\sigma(L) \setminus \sigma_2  = \sigma(S_1).
\]
The same statement for essential spectrum is not explicitly given by \cite[Lemma 2.3.2]{tretter2008spectral} but follows by the same proof. This means
\[
\sigma_{ess}(L) \setminus \sigma_{ess}(L_{11}) = \sigma_{ess}(S_2),
\quad
\sigma_{ess}(L) \setminus \sigma_{ess}(L_{22})  = \sigma_{ess}(S_1).
\]
By Proposition \ref{prop:Lij}, $\sigma_{ess}(L_{11}) = \emptyset$ and so in fact we have
\[
\sigma_{ess}(L) = \sigma_{ess}(S_2).
\]
Using identity \eqref{eq:L11inv}, we have
\[
S_1(\lambda) (\pi_1 A_2 \pi_1^* + \beta I)^{-1} = I + (\pi_1 F(\lambda) \pi_1^* - \beta I + L_{12}(\lambda) L_{22}^{-1}(\lambda) L_{21}(\lambda)) (\pi_1 A_2 \pi_1^* + \beta I)^{-1}
\]
and so since $L_{12}(\lambda) L_{22}^{-1}(\lambda) L_{21}(\lambda)$ is bounded for $\lambda \in \rho_2$ we obtain that, as in \eqref{eq:L11inv}, $S_1(\lambda)$ is Fredholm with index 0 having only eigenvalues with finite multiplicity in its spectrum. We summarise the main results just proven in the following proposition.

\begin{proposition} \label{prop:sigma12}
The spectrum of $\sigma(S_1)$ is discrete:
\[
   \sigma(S_1) \subseteq \sigma_{\mathrm{disc}}(L) ,
\]
where $\sigma_{\mathrm{disc}}(L)$ denotes the discrete component of
$\sigma(L)$. Furthermore,
\[
\sigma_{ess}(L) = \sigma_{ess}(S_2).
\]
\end{proposition}

\begin{remark}
There are eigenvalues of $L$ that do not lie in $\sigma_1$. Specifically, the quasi-rigid modes form three separate
two-dimensional eigenspaces with eigenvalues $\pm \ii |\Omega|$ and
$0$. These eigenvalues are embedded in the essential spectrum. For the sake of self-containedness, the detailed computations are given in Appendix~\ref{A:r}.
\end{remark}

\begin{remark}[Geostrophic modes] \label{rem:geo}
    For completeness of the characterization, we briefly present how the geostropic modes (see \cite[Section 4.1.6]{dahlen2020theoretical}) appear in the analysis. Fluid motions which travel along the level surfaces of $\rho_0$ and preserve the density are eigenfunctions of $L$, or geostrophic modes, corresponding to $\lambda = 0$. %Waves traveling perpendicular to the axis of rotation have zero frequency and are sometimes called the geostrophic modes.
    %Geostrophic (steady) modes are associated to eigenvalue $\lambda=0$ and
    They are necessarily solutions to the problem
\begin{equation} \label{eq:geos0}
\left\{\begin{array}{ll}
&\tilde{s}\cdot u=0,\\
&\nabla\cdot (\rho_0 u)=0,\\
&\nabla\cdot u |_{\partial M}=0.
\end{array}
\right.
\end{equation}
Note that if $u \in H$ satisfies \eqref{eq:geos0}, then $u \in H_2  = \operatorname{Ker}(T)$. If  $\vp\in H^1(M)$ is such that \begin{equation} \label{eq:geos1}
\nabla\vp\cdot(\nabla\times \tilde{s})=0
\end{equation}
and we define $u$ by
\begin{equation}\label{eq:geos}
u=\rho_0^{-1}\nabla\vp\times\tilde{s},
\end{equation}
then $u$ satisfies the first and the second equations of \eqref{eq:geos0}
as $\nabla \cdot(\nabla\vp\times\tilde{s})=\tilde{s}\cdot(\nabla\times\nabla\vp)-\nabla\vp\cdot(\nabla\times \tilde{s}).$ Since we have also $$\nabla\cdot (\rho_0^{-1}\nabla\vp\times\tilde{s})=(\nabla\vp\times\tilde{s})\cdot \nabla\rho_0^{-1},$$
the boundary condition in \eqref{eq:geos0} is equivalent to
$$(\nabla\vp\times\tilde{s})\cdot \nabla\rho_0^{-1} |_{\partial M}= 0.$$
Assuming that $g_0'$, $\nabla \rho_0$ and $n$ are parallel on $\partial M$, which is required for well-posedness of the system, this boundary condition is automatically satisfied. In conclusion, the geostrophic modes form a infinite-dimensional subspace of $H_2.$ This is consistent with the fact that the essential spectrum of $L$ corresponds with the $H_2$ component (i.e. $L(0)$ fails to be Fredholm because of an infinite dimensional kernel contained in $H_2$).

%\textcolor{red}{If we are assuming $\nabla \rho_0$ and $g_0'$ are parallel with $n$ at the boundary, then it appears to me the boundary condition is always satisfied, which makes sense if we want infinitely many such solutions. Are the equations at the top the defining equations for geostropic modes? If so, it is obvious they are in $H_2$ and I guess the conclusion is that infinitely many of them exist.}
\end{remark}

\section{Riesz projectors and acoustic mode decomposition}

A common approach to solution of \eqref{eq: MomentumConservation6} is to expand $u$ in so-called acoustic modes. In practice, this typically means expansion in the eigenfunctions of $L_{11}$, which by Proposition \ref{prop:Lij} correspond to discrete eigenvalues. Indeed, applying the spectral theory on Krein spaces (\cite{langer2008spectral,azizov1981linear}) and properties from Proposition \ref{prop:Lij} it is possible to develop a resolution of the identity for $L_{11}$ using its eigenfunctions. However, these eigenfunctions are not modes for the operator $L$. Indeed, suppose $\lambda \in \rho_2$, which by Proposition \ref{prop:sigma12} is outside of the essential spectrum of $L$. Using the decomposition \eqref{eq:Schur2}, we have that
\[
\mathcal{L}(\lambda)
\begin{pmatrix}
u \\ v
\end{pmatrix} = 0
\]
if and only if
\[
S_1(\lambda) u = 0, \quad L_{22}^{-1}(\lambda) L_{21}(\lambda) u = - v.
\]
Thus, eigenvalues of $L$ outside the essential spectrum and their corresponding modes actually correspond to eigenfunctions of $S_1$, and contain a component in $\operatorname{Ker}(T)$. Therefore, to develop a true expansion for $L$, at least away from the essential spectrum, we should use the eigenfunctions of $S_1$.
% An interesting avenue for future work will be investigation of the relationship between eigenvalues and eigenfunctions of $L_{11}$ and $S_1$.

% As mentioned in Proposition \ref{prop:sigma12}, $\sigma(S_1)$ is discrete, and is identified with the acoustic spectrum. and $S_1$ has the properties to allow resolution of the identity. However, the eigenfunctions of $S_1$ are not the eigenfunctions of $L$ directly. Indeed, suppose that $\lambda \in \rho_2$. Using the decomposition \eqref{eq:Schur2}, we have that
% \[
% \mathcal{L}(\lambda)
% \begin{pmatrix}
% u \\ v
% \end{pmatrix} = 0
% \]
% if and only if
% \[
% S_1(\lambda) u = 0, \quad L_{22}^{-1}(\lambda) L_{21}(\lambda) u = - v.
% \]
% Therefore acoustic modes, which are eigenfunctions of $S_1(\lambda)$ and so contained in $\operatorname{Ran}(T^*)$, give rise to modes of the full system which also have a component in $\operatorname{Ker}(T)$. Using these we can write down an expansion of $L$ as an acoustic mode summation, although this will not include the contributions from other components of the spectrum.

To arrive at an expansion using the proper acoustic modes, we assume secular stability. Then $\gamma(a_2) = 0$ and $A_2^{1/2}$ is well defined on $D(A_2)$, with a nontrivial $\operatorname{Ker}(A_2^{1/2})$ coinciding with $\operatorname{Ker}(A_2)$. We let
% {\color{red} [shouldn't $B_2$ first be defined on $D(a_2) \times (D(a_2)$ and then on a quotient space ($\hat{A}_2$ replaced by $A_2$)? However, in Bognar, it is used that $A_2$ is positive and invertible - hence, the introduction of the quotient space might be needed from the start.]}
\begin{equation}
   B_2 = \begin{pmatrix} 0 & \ii A_2^{1/2} \\
         \ii A_2^{1/2} & -2 R_{\Omega} \end{pmatrix} ,
\end{equation}
with
\[
   D(B_2) = D(a_2) \times D(a_2) .
\]
It is immediate that $\ii B_2$ is \textit{self adjoint} on $H \times H$, equipped with the original inner product; indeed,
\begin{equation}
   \left(B_2 \begin{pmatrix} u \\ v \end{pmatrix} ,
                     \begin{pmatrix} u' \\ v' \end{pmatrix}\right)
   = (\ii A_2^{1/2} v, u')_H
          + (\ii A_2^{1/2} u - 2 R_{\Omega} v, v')_H
   = -\left(\begin{pmatrix} u \\ v \end{pmatrix} ,
        B_2  \begin{pmatrix} u' \\ v' \end{pmatrix}\right) .
\end{equation}
We introduce (noting the minus sign)
\begin{equation}
   \widetilde{L}(\lambda) = B_2 - \lambda \Id
   = \begin{pmatrix} -\lambda & \ii {A}_2^{1/2} \\
         \ii {A}_2^{1/2} & -\lambda-2{R}_{\Omega} \end{pmatrix}
\quad\text{and}\quad
   \widetilde{R}(\lambda) = \widetilde{L}(\lambda)^{-1} . 
\end{equation}
From Remark~\ref{rem:geo}, it follows that $0 \notin \rho(L)$, so for $\lambda \in \rho(L)$ we can invert the previous equation to obtain
\begin{equation}\label{eq:res}
\widetilde{R}(\lambda) = \widetilde{L}(\lambda)^{-1} = \begin{pmatrix} -\lambda^{-1}(\Id - A_2^{1/2} R(\lambda) A_2^{1/2}) & -\ii {A}_2^{1/2} R(\lambda) \\
         -\ii R(\lambda) {A}_2^{1/2} & -\lambda R(\lambda) \end{pmatrix} .
\end{equation}
On the other hand, if $\lambda \in \rho(\widetilde{L})$ we have an inverse
\begin{equation}\label{eq:tR}
\widetilde{R}(\lambda) = \begin{pmatrix}
    \widetilde{R}_{11}(\lambda) & \widetilde{R}_{12}(\lambda) \\
    \widetilde{R}_{12}(\lambda) &
    \widetilde{R}_{22}(\lambda)
\end{pmatrix} .
\end{equation}
Thus the resolvents are related. If $\lambda \neq 0$, $R(\lambda) = -\lambda^{-1} R_{22}(\lambda)$. Remark~\ref{rem:geo} also implies that $0 \notin \rho(\widetilde{L})$ and so we see that $\rho(L) = \rho(\widetilde{L})$. Hence, the spectra are the same.

Suppose that $\lambda \in \sigma_{disc}(L) = \sigma_{disc}(\widetilde{L})$. Then a corresponding eigenfunction $(u, v) \in H \times H$ will satisfy
\[
   \ii A_2^{1/2} v = \lambda u ,\quad
   \ii A_2^{1/2} u - 2 R_\Omega v = \lambda v .
\]
Restricting to acoustic modes, $v = 0$ is not possible since $\lambda \neq 0$. (That is, $\lambda = 0$ is an eigenvalue but does not correspond with an acoustic mode.) Thus we can combine these formulae to obtain
\begin{equation}\label{eq:Lproj}
   L(\lambda) v = 0 ,\quad
   u = \lambda^{-1} \ii A_2^{1/2} v .
\end{equation}
Using the above calculations, we can introduce the Riesz projectors onto the space of acoustic modes, which are the spectrum of $S_1$. Indeed, let $\lambda \in \sigma(S_1)$ and $\Gamma_\lambda$ be a contour around surrounding $\lambda$ and no other part of $\sigma(B_2)$. Then consider the standard formula for the projection onto the eigenspace of $\lambda$
\[
\widetilde{P}_\lambda = \frac{1}{2\pi i} \oint_{\Gamma_{\lambda}} \widetilde{R}(\mu) \ \mathrm{d} \mu.
\]
For more information on the definition of the Riesz projection and this contour integral, see \cite{hislop2012introduction}. we further let $\pi_v$ be projection onto the $v$ component and define $P_\lambda = \pi_v \widetilde{P}_\lambda \pi^*_v$. Then, using \eqref{eq:res},
\[
P_\lambda = -\frac{\lambda}{2\pi i} \oint_{\Gamma_\lambda} R(\mu) \ \mathrm{d} \mu.
\]
We can now use these projectors to define the projection onto the acoustic part of the spectrum, which is
\[
   E = \sum_{\lambda \in \sigma(S_1)} P_{\lambda}.
\]
We conclude that the projection onto the eigenspace of $\lambda$ for $\widetilde{L}$ gives a corresponding projection, by taking the $v$ component as in \eqref{eq:Lproj}, onto the space $\operatorname{Ker}(L(\lambda))$ of an acoustic mode. This projection $E$ shows it is possible to express the acoustic part of the wavefield as a sum of normal modes. %For $\lambda = 0$, the (non-acoustic) eigenspace for $\widetilde{L}$ is $\operatorname{Ker}(B_2)$. This is a Jordan chain of the quadratic pencil $L(\lambda)$ to the zero eigenvalue [to be checked].}

Using the above mentioned Riesz projectors, we obtain a partial spectral decomposition of $\widetilde{R}_{22}(\lambda)$, namely into acoustic modes:
\[
\widetilde{R}_{22}(\lambda)|_{acoustic} = \sum_{\omega \in \sigma(S_1)} \frac{P_\omega}{(\omega- \lambda)}.
\]
This induces a corresponding partial spectral decomposition of $R(\lambda)$ from \eqref{eq:res}:
\[
R(\lambda)|_{acoustic} = \frac{1}{\lambda}\sum_{\omega \in \sigma(S_1)} \frac{P_\omega}{(\lambda- \omega)},
\]
which is commonly used in computations.

\section{Inertia-gravity modes and essential spectrum}
\label{sec:pure-point-dense}

We now investigate the essential spectrum of $L$. Because $L_{11}^{-1}(\lambda)$ is compact and the $L_{ij}(\lambda)$ are bounded from Proposition \ref{prop:Lij}, using Proposition \ref{prop:sigma12} we have that
\[
\sigma_{ess}(L) = \sigma_{ess}(L_{22}).
\]
Using the formula for $L_{22}$ given in Remark \ref{rem:Lij} and Lemma \ref{lem:-4-}, this further reduces to
\begin{equation} \label{eq:esseq}
\sigma_{ess}(L) = \sigma_{ess}\left (\pi_2 \left ( F(\lambda) + N^2 \hat{g}_0' \hat{g}_0'^T \right ) \pi_2^* \right ).
\end{equation}
Thus, referring to \eqref{eq:Fldef}, we are led to consider the spectrum of
\[
M(\lambda) = \pi_2 (\lambda^2 \rm{Id}+2\lambda R_\Omega + N^2 \hat{g}_0' \hat{g}_0'^T) \pi_2^*: \operatorname{Ker(T)} \rightarrow \operatorname{Ker(T)}.
\]
Solutions $u \in \operatorname{Ker}(T)$ of
\begin{equation} \label{eq:igm_1}
\partial_t^2 u + 2 \Omega \times \partial_t u + N^2 \hat{g}_0' \hat{g}_0'^T u = 0
\end{equation}
are modes of $M$, referred to as inertia-gravity modes. Indeed, the restoring force of inertial modes is the Coriolis force, $2 \Omega \times \partial_t (\rho_0 u)$, while the restoring force of gravity modes is the buoyancy, $(\nabla \cdot \rho_0 u) g_0'$, which equals $N^2 \hat{g}_0' \hat{g}_0'^T \rho_0 u$ for $u \in \operatorname{Ker}(T)$. With both restoring forces, we speak of inertia-gravity modes.

We will precisely characterise the essential spectrum of $L$ in Theorem \ref{thm:pp}, but must first introduce some notation and a definition. For convenience, let us define $P_\xi^\perp = \sigma_p(\pi_2)$ defined by \eqref{eq:pisymbol} which is the projection onto the space orthogonal to $\xi$. Also, for $\Omega \in \mathbb{R}^3$ let $\Omega_{\xi}$ be the component of $\Omega$ in the direction $\xi$ given by
\[
\Omega_{\xi} = \frac{\xi \cdot \Omega}{|\xi|}.
\]

\begin{definition} \label{def:sigmapt}
    For $x \in M$ and $\xi \in \mathbb{R}^3 \setminus \{0\}$, let $\sigma_{pt}(x,\xi)$ be the set of $\lambda \in \mathbb{C}$ such that
    \begin{equation}\label{eq:interiorp}
\mathbb{C}^3 \ni \eta \mapsto \lambda^2 P_\xi^\perp \eta + 2 \lambda P_\xi^\perp( \Omega \times P_\xi^\perp \eta) + N^2 (\hat{g}_0 \cdot P_\xi^\perp \eta) P_\xi^\perp \hat{g}_0.
\end{equation}
has rank less than two (note that two is the largest possible rank due to $P_\xi^\perp$).
\end{definition}

\noindent In fact, the set $\sigma_{pt}(x,\xi)$ can be precisely characterised, which is done in the next lemma.

\begin{lemma}\label{lem:pt}
If $\lambda \in \sigma_{pt}(x,\xi)$, then $\lambda = 0$ or
\begin{equation}\label{eq:lam1}
\lambda = \pm i \sqrt{4 \Omega_\xi^2 + N^2 |P_\xi^\perp \hat{g}_0|^2}.
\end{equation}
\end{lemma}

\begin{proof}
First, assume that $P_\xi^\perp \hat{g}_0 \neq 0$ and set
\begin{equation}\label{eq:u1}
\eta = a\ P_\xi^\perp \hat{g}_0 + b\ \xi \times P_\xi^\perp \hat{g}_0 
\end{equation}
where $a$ and $b$ are constants, not both equal to zero, to be determined. Calculation shows
\[
\begin{split}
P_\xi^\perp \Big (\Omega \times P_\xi^\perp(\xi \times P_\xi^\perp \hat{g}_0 )\Big ) & = P_\xi^\perp \Big (\Omega \times (\xi \times P_\xi^\perp \hat{g}_0 ) \Big )\\
& = - |\xi| \Omega_\xi P_\xi^\perp \hat{g}_0
\end{split}
\]
and
\[
P_\xi^\perp ( \Omega \times P_\xi^\perp \hat{g}_0) = \frac{\Omega_\xi}{|\xi|} \xi \times P_\xi^\perp \hat{g}_0.
\]
Therefore, if $\lambda \in \sigma_{pt}(x,\xi)$ then for some $a$ and $b$
\[
\Big ( \lambda^2 a - 2 \lambda |\xi|\Omega_\xi b + N^2 |P_\xi^\perp \hat{g}_0|^2 a \Big ) P_\xi \hat{g}_0 + \left ( \lambda^2 b + 2 \lambda \frac{\Omega_\xi}{|\xi|}a \right ) \xi \times P_\xi^\perp \hat{g}_0 = 0.
\]
Setting the two coefficients equal to zero, we see that either $\lambda = 0$ and $a = 0$ or
\begin{equation}\label{eq:lam2}
\lambda^2 = -4 \Omega_\xi^2 - N^2 |P_\xi^\perp \hat{g}_0|^2
\end{equation}
which completes the proof in this case. When $P_\xi^\perp\hat{g}_0 = 0$, we choose arbitrary $w$ orthogonal to $\xi$ and start with
\[
\eta = a\ w + b\ \xi \times w
\]
instead of \eqref{eq:u1}. A similar calculation gives $\lambda = 0$ or \eqref{eq:lam2} in this case, and so the lemma is proven.
\end{proof}

If $\lambda$ satisfies \eqref{eq:lam1}, then
\begin{equation} \label{eq:lambsq}
\lambda^2 = - \frac{1}{|\xi|^2} \Big (4 (\Omega\cdot \xi)^2 + N^2 |\xi|^2 - N^2 (\hat{g}_0' \cdot \xi)^2\Big ).
\end{equation}
The quantity in parentheses above is a quadratic form in $\xi$, and by determining the eigenvalues of the corresponding matrix we can determine the range of possible values of $\lambda^2$. These eigenvalues are $N^2$ and
\begin{equation}\label{eq:betapm}
\beta_\pm = \frac{1}{2} \left ( 4 |\Omega|^2 + N^2 \pm \sqrt{(N^2 + 4 |\Omega|^2)^2 - 16 (\Omega \cdot \hat{g}_0')^2 N^2} \right ).
\end{equation}
Therefore, the range of possible $\lambda^2$ will be $-1$ times the interval between the minimum and maximum of these eigenvalues. If $N^2\geq 0$, then this range will be $\lambda^2 \in -[\beta_-,\beta_+]$ which leads to $\lambda \in \pm i [\sqrt{\beta_-},\sqrt{\beta_+}]$. This agrees with the range of non-ellipticity of the Poincar\'e operator determined in \cite{VidalCdV:2024}. In this case, it will be useful later for the proof Lemma \ref{lem:Lopatinskii} to note that $\sqrt{N^2} \in [\sqrt{\beta_-},\sqrt{\beta_+}]$. If $N^2 < 0$, then the range of possible values is $\lambda^2 \in -[N^2,\beta_+]$, which gives $
\lambda \in [-\sqrt{(-N^2)},\sqrt{(-N^2)}] \cup i [-\sqrt{\beta_+},\sqrt{\beta_+}]$. We combine these cases in the next lemma.

%\tred{From \eqref{eq:betapm}, in the case $N^2> 0$ we note that $\beta_- = 0$ exactly when $\Omega\cdot \hat{g}_0' = 0.$ Since $\Omega$ is a fixed vector, it is highly likely that this equality occurs at some point $x$ within the planet. At $x$, we have $\lambda \in i[-\sqrt{\beta_+},\sqrt{\beta_-}]$ }

\begin{lemma}
    Let $\beta_\pm$ be given by \eqref{eq:betapm}. Then
    \begin{equation}\label{eq:spt}
    \bigcup_{\xi \in \mathbb{R}^3 \setminus \{0\}} \sigma_{pt}(x,\xi) = \bigcup_{\pm \in \{-1,1\}} \Bigg (\left [ -\sqrt{\max(0,-N^2)},\sqrt{\max(0,-N^2)}\right ]\cup \pm i\left [ \sqrt{\max(0,\beta_-)},\sqrt{\beta_+} \right ]\Bigg ).
    \end{equation}
   Furthermore, this set contains $\sqrt{-N^2}$.
\end{lemma}

\begin{figure}
    \centering
(a)
\begin{tikzpicture}
\draw[thick,->] (-3,0) -- (3,0) node[anchor=west] {\small$\nu$};
\draw[thick,->] (0,-3) -- (0,3) node[anchor=west] {\small$\omega$};
\draw (-2,2) node {\small$\lambda = \nu + i \omega$};
%\draw[line width=.7mm] (0,-3) -- (0,2.9);
\draw[line width=1mm] (0,1.25) -- (0,2.25);
\draw[line width=1mm] (0,-1.25) -- (0,-2.25);
\draw[line width=1mm,dashed] (0,-.7) -- (0,.7);
%\draw[line width=1mm] (-1.5,0) -- (1.5,0);
%\draw (-1.5,-.2) -- (-1.5,.2);
%\draw (1.5,-.2) -- (1.5,.2);
%\draw (-2,-.2) -- (-2,.2);
%\draw (2,-.2) -- (2,.2);
\draw (-.2,-2.25) -- (0.2,-2.25)
node[anchor=west]
{\small$-\beta_+$};
\draw (-.2,2.25) -- (0.2,2.25)
node[anchor=west]
{\small$\beta_+$};
\draw (-.2,-1.25) -- (0.2,-1.25)
node[anchor=west] {\small$-\beta_-$};
\draw (-.2,1.25) -- (0.2,1.25)
node[anchor=west] {\small$\beta_-$};
\draw (-.2,0.7) -- (0.2,0.7)
node[anchor=west] {\small$|P_n^\perp \hat{g}_0'| \sqrt{N^2}$};
\draw (-.2,-0.7) -- (0.2,-0.7)
node[anchor=west] {\small$-|P_n^\perp \hat{g}_0'| \sqrt{N^2}$};
\draw (0,0) node[circle,fill,inner sep=1mm]{};
\end{tikzpicture}
(b)
\begin{tikzpicture}
\draw[thick,->] (-3,0) -- (3,0) node[anchor=west] {\small$\nu$};
\draw[thick,->] (0,-3) -- (0,3) node[anchor=west] {\small$\omega$};
%\draw[line width=.7mm] (0,-3) -- (0,2.9);
\draw[line width=1mm] (0,-2.25) -- (0,2.25);
\draw[line width=1mm] (-1.5,0) -- (1.5,0);
%\draw[line width=1mm] (-1.5,0) -- (1.5,0);
%\draw (-1.5,-.2) -- (-1.5,.2);
%\draw (1.5,-.2) -- (1.5,.2);
%\draw (-2,-.2) -- (-2,.2);
%\draw (2,-.2) -- (2,.2);
\draw (-.2,-2.25) -- (0.2,-2.25)
node[anchor=west]
{\small$-\beta_+$};
\draw (-.2,2.25) -- (0.2,2.25)
node[anchor=west]
{\small$\beta_+$};
\draw (-1.5,-.2) -- (-1.5,.2)
node[anchor=south]
{\small$-\sqrt{-N^2}$};
\draw (1.5,-.2) -- (1.5,.2)
node[anchor=south]
{\small$\sqrt{-N^2}$};
\end{tikzpicture}
    \caption{The solid black regions give the set \eqref{eq:spt} for fixed $x \in M$; in (a) on the left we see the case when $N^2\geq 0$ (note that this region includes the origin as indicated by the black dot), while in (b) on the right we see the case $N^2<0$. In reference to Theorem \ref{thm:pp}, the solid black regions are the areas where ellipticity fails at the point $x$, and appear for a single $x \in M$ in the union on the top line of \eqref{eq: SpectrumCondition}. The dashed black region on the left is the set where the system \eqref{eq:bigsys} fails the Lopatinskii condition at the boundary but not in the interior, and appears for a single $x \in \partial M$ in the union on the second line of \eqref{eq: SpectrumCondition}. Note that in (a) it is possible for the dashed set to intersect the solid set.}
    \label{fig:pointwise}
\end{figure}
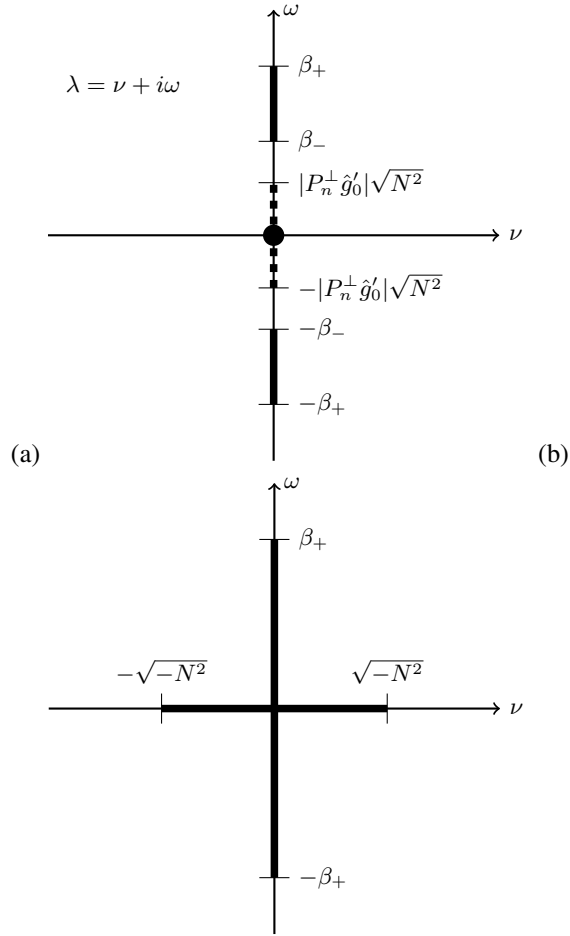

We now present our characterisation of the essential spectrum of $L$, which is the main result of this paper.

\begin{theorem} \label{thm:pp}
%The inertia-gravity modes form the essential spectrum, $\sigma_{ess}(L)$.
For $x\in \partial M$, let $n(x)$ denote the inward pointing unit normal vector. The essential spectrum $\sigma_{ess}(L)$ is given by
\begin{equation}
\begin{split}
   \sigma_{ess}(L) & = \left (\bigcup_{x\in M,\ \pm\in \{-1,1\}} \left [ -\sqrt{\max(0,-N^2)},\sqrt{\max(0,-N^2)}\right ]\cup \pm i\left [ \sqrt{\max(0,\beta_-)},\sqrt{\beta_+} \right ] \right )\\
   & \bigcup \left (\bigcup_{x\in \partial M} i|P_n^\perp \hat{g}_0'| \left [ -\sqrt{\max(0,N^2)},\sqrt{\max(0,N^2)}\right ] \right ).
\end{split}
\label{eq: SpectrumCondition}
\end{equation}
% where
% \begin{equation} \label{def:S1i}
%    \mathfrak{S}_{1,i} = \{ \nu + \ii \omega \in \mathbb{C}\ :\
%    \nu = 0\ \text{and}\
%        \omega^2 \le 4 \Omega^2 + \max(0,N^2_{\mathrm{sup}})\
%    \text{or}\ \omega = 0\ \text{and}\
%        \nu^2 \le \max(0,-N^2_{\mathrm{inf}}) \}
% \end{equation}
% corresponds to modes concentrated in the interior of $M$ {\color{red} [SEAN COMMENT: I've realised there is a slight problem here. In fact the essential spectrum from the interior will in general only be a subset of $\mathfrak{S}_{1,i}$ since the sup and inf of $N^2$ may not occur at the correct places. A fully precise statement would be to make
% \[
% \mathfrak{S}_{1,i} = \bigcup_{(x,\xi)\in M \times \mathbb{R}^3 \setminus \{0\}}\sigma_{pt}(x,\xi)
% \]
% See proof of Lemma 6 for more detail.]},
% \begin{equation} \label{def:S1b}
%    \mathfrak{S}_{1,b} = ??
% \end{equation}
% corresponds to modes concentrated on the boundary $\partial M$, and
% \begin{equation} \label{def:S2}
%    \mathfrak{S}_2 = \{ \nu + \ii \omega \in \mathbb{C}\ :\
%    \nu = 0\ \text{and}\
%        |\omega| \le |\Omega| + [\Omega^2
%                         + \max(0,N^2_{\mathrm{sup}})]^{1/2}\
%    \text{or}\ |\omega| \le |\Omega|\ \text{and}\
%        \omega^2 + \nu^2 \le \max(0,-N^2_{\mathrm{inf}}) \} .
% \end{equation}
\end{theorem}

Before proving Theorem \ref{thm:pp}, we consider some special cases of the set in \eqref{eq:spt} from which we can obtain an upper bound on the essential spectrum in \eqref{eq: SpectrumCondition}. If for some value of $x$ we have $\Omega\cdot \hat{g}_0 = 0$, then from \eqref{eq:betapm} we have
\[
\beta_\pm = \min(0,4 |\Omega|^2+N^2), \ \max(0, 4|\Omega|^2 + N^2).
\]
Also, for general points $\beta_+ \leq 4 |\Omega|^2 + N^2$. Therefore, considering \eqref{eq: SpectrumCondition}, we see that the part of $\sigma_{ess}(L)$ along the imaginary axis must be contained in
\[
i\left [-\sqrt{4 |\Omega|^2 + \max(0,N_{\mathrm{sup}}^2)},\sqrt{4 |\Omega|^2 + \max(0,N_{\mathrm{sup}}^2)} \right ].
\]
On the other hand, directly from \eqref{eq: SpectrumCondition} we see that the part of $\sigma_{ess}(L)$ along the real axis must be contained in
\[
\left [ -\sqrt{\max(0,-N_{\mathrm{inf}}^2)},\sqrt{\max(0,-N_{\mathrm{inf}}^2)}\right ]
\]
Putting the previous remarks together, we see that
\begin{equation} \label{def:S1}
   \sigma_{ess}(L) \subset \mathfrak{S}_{1} = \{ \nu + \ii \omega \in \mathbb{C}\ :\
   \nu = 0\ \text{and}\
       \omega^2 \le 4 |\Omega|^2 + \max(0,N^2_{\mathrm{sup}})\
   \text{or}\ \omega = 0\ \text{and}\
       \nu^2 \le \max(0,-N^2_{\mathrm{inf}}) \}.
\end{equation}
An illustration of the set $\mathfrak{S}_1$ is given in Figure \ref{fig:cross}.
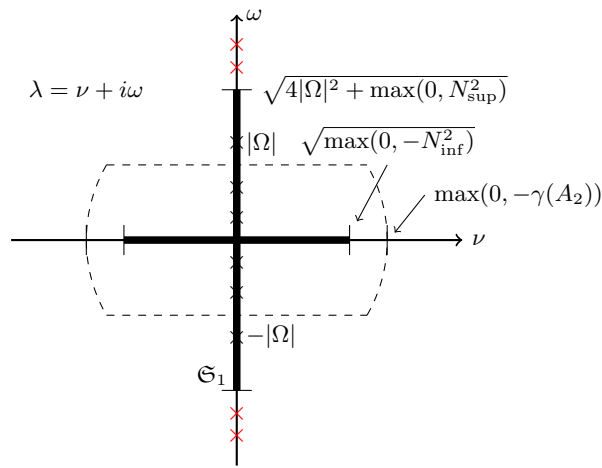
\begin{figure}
    \centering
\begin{tikzpicture}
\draw[thick,->] (-3,0) -- (3,0) node[anchor=west] {\small$\nu$};
\draw[thick,->] (0,-3) -- (0,3) node[anchor=west] {\small$\omega$};
\draw (-2,2) node {\small$\lambda = \nu + i \omega$};
%\draw[line width=.7mm] (0,-3) -- (0,2.9);
\draw [dashed,domain=-30:30] plot ({2*cos(\x)},{2*sin(\x)});
\draw [dashed,domain=150:210] plot ({2*cos(\x)},{2*sin(\x)});
\draw[dashed] ({2*cos(150)},{2*sin(150)}) -- ({2*cos(30)},{2*sin(30)});
\draw[dashed] ({2*cos(210)},{2*sin(210)}) -- ({2*cos(-30)},{2*sin(-30)});
\draw[line width=1mm] (0,-2) -- (0,2);
\draw[line width=1mm] (-1.5,0) -- (1.5,0);
\draw (-1.5,-.2) -- (-1.5,.2);
\draw (1.5,-.2) -- (1.5,.2);
\draw (-2,-.2) -- (-2,.2);
\draw (2,-.2) -- (2,.2);
\draw (-.2,-2) -- (0.2,-2);
\draw (-.2,2) -- (0.2,2) node[anchor=west] {\small$\sqrt{4 |\Omega|^2 + \mathrm{max}(0,N^2_{\sup})}$};
\draw (0,{2*sin(30)}) node[anchor=south west] 
{\small$|\Omega|$};
\draw (0,{-2*sin(30)}) node[anchor=north west] {\small$-|\Omega|$};
\draw[->] (2,1) node[anchor=south] {\small$\sqrt{\mathrm{max}(0,-N^2_{\inf})}$} -- (1.6,.25);
\draw[->] (2.5,.6) node[anchor=west] {\small$\mathrm{max}(0,-\gamma(A_2))$} -- (2.1,.1);
\draw (0,-1.8) node[anchor=east] {\small$\mathfrak{S}_1$};
\draw (0,.3) node {$\times$};
\draw (0,-.3) node {$\times$};
\draw (0,.7) node {$\times$};
\draw (0,-.7) node {$\times$};
\draw (0,1.3) node {$\times$};
\draw (0,-1.3) node {$\times$};
\draw (0,2.3) node[color=red] {$\times$};
\draw (0,-2.3) node[color=red] {$\times$};
\draw (0,2.6) node[color=red] {$\times$};
\draw (0,-2.6) node[color=red] {$\times$};
\end{tikzpicture}
\caption{An illustration of the spectrum $\sigma(L)$ after Rogister and Valette \cite{RogisterValette:2009}. The dark cross is the set $\mathfrak{S}_1$ which by Theorem \ref{thm:pp}, Lemma \ref{lem:pt} and the discussion after the proof of that lemma, we have shown must contain the essential spectrum $\sigma_{ess}(L)$, but may in general be larger than the essential spectrum. By Proposition \ref{prop:DS}, the full spectrum $\sigma(L)$ is contained in the union of the imaginary axis and region surrounded by the dashed curve. The crosses on the imaginary axis are included to indicate eigenvalues, which could also occur within the dashed curve. The crosses which appear outside of the essential spectrum are red, indicating that they are part of $\sigma(S_1)$ which is the acoustic part of the spectrum.}

\label{fig:cross}
\end{figure}

\tred{In fact, the inclusion in \eqref{def:S1} will always be an equality for the rotating self-gravitating truncated gas planets that we consider. This is because, using \eqref{eq: PerturbGrav}, calculation shows
\[
\Omega \cdot g_0'(x) = -\Omega \cdot \int_{M} G\frac{(x - x')}{|x-x'|^3} \rho^0(x')\, \mathrm{d} x'.
\]
Let $x_{max} \in \partial M$ maximize $\Omega \cdot x$. Then for any $x' \in \partial M$, $\Omega \cdot (x_{max}-x')\geq 0$ which implies $\Omega \cdot g_0'(x_{max}) \leq 0$. Similarly, by choosing $x_{min}$ that minimizes $\Omega \cdot x$ we can show $\Omega \cdot g_0'(x_{min}) \geq 0$. Therefore, by continuity, $\Omega \cdot \hat{g}'_0 = 0$ somewhere in $M$ and so at this point $\beta_- = 0$. Because of this, we can conclude equality in \eqref{def:S1}.
%In the spherically symmetric case when $N^2$ is radial (while omitting the centrifugal potential), then, considering the arguments above, the inclusion in \eqref{def:S1} will be an equality.
On the other hand, in the f-plane approximation as considered in \cite{VidalCdV:2024}, $\Omega \cdot \hat{g}_0'$ is constant and so never vanishes and \eqref{def:S1} is a proper inclusion.}

\begin{remark}
For a neutrally buoyant planet, $\tilde{s} = 0$ and $N^2 = 0$. Then the relevant operator $M(\lambda)$ reduces to the Poincar\'{e} operator \cite{RieutordNoui:1999}. In the polytropic model, the planet is neutrally buoyant. In the case that $N^2 = 0$, Theorem \ref{thm:pp} also gives the essential spectrum of $2 \pi_2 R_\Omega \pi_2^*$ which, by Lemma \ref{lem:pt}, is $i[-|\Omega|,|\Omega|]$. As observed in \cite[Page 138]{valette1987spectre}, this interval contains the spectrum of $2 \pi_2 R_\Omega\pi_2^*$, and so must in fact be equal to the full spectrum.
\end{remark}

\begin{proof}[Proof of Theorem \ref{thm:pp}]
    We begin by proving the inclusion,
\begin{equation}\label{eq:leftinc}
\begin{split}
    \sigma_{ess}(L) &\supset \bigcup_{x\in M,\ \pm\in \{-1,1\}} \left [ -\sqrt{\max(0,-N^2)},\sqrt{\max(0,-N^2)}\right ]\cup \pm i\left [ \sqrt{\max(0,\beta_-)},\sqrt{\beta_+} \right ]\\
    & = \bigcup_{(x,\xi)\in M\times \mathbb{R}^3 \setminus\{0\}} \sigma_{pt}(x,\xi).
\end{split}    
\end{equation}
    Our method for this step is inspired by \cite[Theorem 2.1]{CdV:2020}, which considers a similar but simpler problem for a scalar function. Suppose that $\lambda \in \mathbb{C}$ is contained in $\sigma_{pt}(x_0,\xi_0)$ such that $x_0 \in M^{int}$. Thus, there exists nonzero $\eta$ orthogonal to $\xi_0$ such that
\begin{equation}\label{1.1}
    \lambda^2 P_{\xi_0} \eta + 2 \lambda P_{\xi} (\Omega \times P_{\xi_0} \eta) + N^2 (P_{\xi_0} \hat{g}_0 \cdot P_{\xi_0} \eta) P_{\xi_0} \hat{g}_0 = 0.
\end{equation}
Then, for any $\epsilon>0$, choose a neighbourhood $U \subset M^{int}$ of $x_0$ such that at all $x \in U$
\[
|\lambda^2 P_{\xi_0} \eta + 2 \lambda P_{\xi_0} (\Omega \times P_{\xi_0} \eta) + N^2 (P_{\xi_0} \hat{g}_0 \cdot P_{\xi_0} \eta) P_{\xi_0} \hat{g}_0| < \epsilon.
\]
Let $\phi \in C_c^\infty(U)$ be such that $\|\phi\|_{L^2(\rho^0  \dd x)} = 1$ and consider
\[
u(x) = \eta \phi(x) e^{it x \cdot \xi_0}.
\]
Considering the Fourier transform, we can see that as $t \rightarrow \infty$, $u$ converges to zero weakly. Since $\pi_2$ is a pseudodifferential operator with principal symbol given by \eqref{eq:pisymbol}, using the fact that $\xi_0$ is orthogonal to $\eta$, we have
\[
\pi_2(u)(x) = \phi(x) e^{it x \cdot \xi_0} \eta + O\left (\frac{1}{t} \right ).
\]
Therefore, for $t$ sufficiently large $\|\pi_2(u)\|_{L^2(\rho_0 \dd x)^3}>C>0$ where $C$ is a constant independent of $t$. Also, since $\pi_2$ is continuous $\pi_2(u)$ converges weakly to zero as $t \rightarrow \infty$. Let us set
\[
v = \frac{\pi_2(u)}{\| \pi_2(u) \|_{H}} \in \operatorname{Ker}(T).
\]
Then
\[
M(\lambda)v = \frac{1}{\| \pi_2(u) \|_{H}}\pi_2 (\lambda^2 {\rm Id}+2\lambda R_\Omega + N^2 \hat{g}_0' \hat{g}_0'^T) \pi_2 u
\]
and the operator on the right side is a pseudodifferential operator with principal symbol given by the map \eqref{eq:interiorp}. Thus,
\[
M(\lambda)v = \frac{1}{\| \pi_2(u) \|_{H}}\Big (\lambda^2 P_{\xi_0} \eta + 2 \lambda P_{\xi_0} (\Omega \times P_{\xi_0} \eta) + N^2 (P_{\xi_0} \hat{g}_0 \cdot P_{\xi_0} \eta) P_{\xi_0} \hat{g}_0 \Big ) \phi(x) e^{it x \cdot \xi_0} + O\left (\frac{1}{t} \right )
\]
and so by taking $t$ sufficiently large
\[
\|M(\lambda) v \|_{L^2(\rho_0 \dd x)^3} \leq \frac{2}{\| \pi_2(u) \|_{H}} \epsilon.
\]
Since $\epsilon >0$ was arbitrary we see that $M(\lambda) v$ converges to zero strongly and so $v$ defines a Weyl sequence. Therefore $\lambda \in \sigma_{ess}(M) = \sigma_{ess}(L)$. This proves $\sigma_{pt}(x_0,\xi_0) \subset \sigma_{ess}(L)$ for $x_0 \in M^{int}$. Since the essential spectrum is closed and \eqref{eq:lam1} is a continuous function of $x$ once $\pm$ is chosen, for $x_0 \in \partial M$ we can take a limit from $M^{int}$ to show $\sigma_{pt}(x_0,\xi) \subset \sigma_{ess}(L)$. This completes the proof of \eqref{eq:leftinc}.

% To prove the right inclusion of \eqref{eq: SpectrumCondition}, we will establish
% \begin{equation}\label{eq:rightinc}
%     \sigma_{ess}(L)^c \supset \left (\bigcup_{(x,\xi)\in M \times \mathbb{R}^3 \setminus \{0\}}\sigma_{pt}(x,\xi) \right )^c.
% \end{equation}
To complete the proof, our method will be to introduce a certain system of PDEs, then show that this system satisfies the Lopatinskii conditions \cite{agmon1964estimates}, see also \cite[Chapter 5, Proposition 11.9]{Taylor1}, if and only if $\lambda$ is in the complement of the right side of \eqref{eq: SpectrumCondition}. When the Lopatinskii conditions are satisfied, the system is a Fredholm operator which implies $M(\lambda)$ is also Fredholm. Therefore, in this case $\lambda \in \sigma_{ess}(L)^c$ which will establish the right inclusion of \eqref{eq: SpectrumCondition}. The Lopatinskii conditions fail if either the system is not elliptic in the interior, or at the boundary. As we will see, interior ellipticity of the system is equivalent to
\begin{equation}\label{eq:lamintell}
\lambda \in \left (\bigcup_{(x,\xi)\in M\times \mathbb{R}^3 \setminus\{0\}} \sigma_{pt}(x,\xi) \right )^c.
\end{equation}
We have already shown that failure of this condition leads to existence of a Weyl sequence. Assuming interior ellipticity, we will show that boundary ellipticity is equivalent to
\[
\lambda \in \left (\bigcup_{x\in \partial M} i |P_n^\perp \hat{g}_0'|\left [ -\sqrt{\max(0,N^2)},\sqrt{\max(0,N^2)}\right 
]\right )^c
\]
We will show that failure of this condition also leads to existence of a Weyl sequence, which will complete the proof. Let us begin now deriving the PDE system.

For any $v \in H$ let us consider the decomposition given by Lemma \ref{lem:T-compact}, which can be written as
\[
v = w + T^* \varphi
\]
where $w \in \operatorname{Ker}(T)$ and $\varphi \in H^1(M)$. Let us further decompose $w$ according the standard Helmholtz decomposition as
\[
w = \nabla \times (\rho_0 w_v) + \nabla \varphi_v
\]
where $\varphi_v \in H^1(M)$ and the vector potential $\rho_0 w_v$ is in the space
\[
H_{\mathrm{Curl},0}(M) = \{ u \in L^2(\rho_0 \dd x) \ : \ \nabla \times u \in L^2(\rho_0 \dd x), \ n\times u|_{\partial M} = 0 \},
\]
while also satisfying
\[
\nabla \cdot (\rho_0 w_v) = 0.
\]
Given that $M$ is a ball, a unique such decomposition exists (see \cite[Section 3]{alberti2019essential}). Let us set $\rho_0 z_v = \nabla \varphi_v$ which must then satisfy
\[
\nabla \times (\rho_0 z_v) = 0.
\]
Then $w \in \operatorname{Ker}(T)$ is equivalent to
\[
\nabla \cdot (\rho_0 z_v) + \frac{g_0'}{c^2} \cdot \nabla \times (\rho_0 w_v) + \frac{\rho_0 g_0'}{c^2} \cdot z_v = 0 , \quad n \cdot z_v|_{\partial M} = 0.
\]
Now, suppose that $u \in \operatorname{Ker}(T)$ satisfies
\begin{equation}\label{eq:Mluf}
M(\lambda) u = f.
\end{equation}
As described above for $v$, there will be $w_u$ and $z_u$ such that
\[
u = \nabla \times (\rho_0 w_u) + \rho_0 z_u
\]
where
\begin{align}
\nabla \times (\rho_0 z_u) & = 0, \label{eq:curlz}
\\ \nabla \cdot (\rho_0 w_u) & = 0, \\
\nabla \cdot (\rho_0 z_u) + \frac{g_0'}{c^2} \cdot \nabla \times (\rho_0 w_u) + \frac{\rho_0  g_0'}{c^2} \cdot z_u & = 0 , \\ n \cdot z_u|_{\partial M} &= 0, \\ n \times w_u|_{\partial M} &= 0. \label{eq:wbd}
\end{align}
We comment that the same equations \eqref{eq:curlz}-\eqref{eq:wbd} will hold for $w_v$ and $z_v$ constructed above for arbitrary $v$. Indeed, let $V(\lambda) = \lambda^2 \mathrm{I} + 2 \lambda R_\Omega + N^2 \hat{g}_0'\hat{g}_0'^T$ and
\[
v = V(\lambda) u
\]
so that $f = \pi_2 v$. These equations become
\begin{align}
\nabla \times (\rho_0 w_v) + \rho_0 z_v + T^* \varphi_v & = V(\lambda) ( \nabla \times (\rho_0 w_u) + \rho_0 z_u), \label{eq:vVu}\\
f & = \nabla \times(\rho_0  w_v) + \rho_0 z_v.
\end{align}
To make the system of equations elliptic, we will also add several potentials $\psi_u$, $\psi_v$ and $\widetilde{\varphi}$. Setting these equal to zero, we find that the following system is satisfied.
\begin{equation} \label{eq:bigsys}
\begin{pmatrix}
    \frac{g_0'^T}{c^2} \nabla \times \rho_0 & \nabla \cdot \rho_0 + \frac{\rho_0 g_0'^T}{c^2} & 0 & 0 & 0 & 0 & 0 & 0\\
    0 & \nabla \times \rho_0  & \nabla \rho_0 & 0 & 0 & 0 & 0 & 0\\
    \nabla \cdot \rho_0 & 0 & 0 & 0 & 0 & 0 & 0 & 0\\
    0 & 0 & 0 & \frac{g_0'^T}{c^2}\nabla \times \rho_0  & \nabla \cdot\rho_0  + \frac{\rho_0 g_0'^T}{c^2} & 0 & 0 & 0\\
    0 & 0 & 0 & 0 & \nabla \times\rho_0  & \nabla\rho_0  & 0 & 0\\
    0 & 0 & 0 & \nabla \cdot\rho_0  & 0 & 0 & 0 & 0\\
    V(\lambda) \nabla \times\rho_0  & V(\lambda)\rho_0  & 0 & - \nabla \times\rho_0 & -\rho_0 I & 0 & -T^* & 0\\
    0 & 0 & 0 & \nabla \times \rho_0  & \rho_0 I & 0 & 0 & -T^*
\end{pmatrix}
\begin{pmatrix}
    w_u\\
    z_u\\
    \psi_u\\
    w_v\\
    z_v\\
    \psi_v\\
    \varphi_v\\
    \widetilde{\varphi}
\end{pmatrix}
=
\begin{pmatrix}
    0\\
    0\\
    0\\
    0\\
    0\\
    0\\
    0\\
    f\\
\end{pmatrix},
\end{equation}
\[
n\times w_u|_{\partial M} = n\times w_v|_{\partial M} = 0,\ n\cdot z_u|_{\partial M} = n \cdot z_v|_{\partial M} = \psi_u|_{\partial M} = \psi_v|_{\partial M} = 0.
\]
In Lemma \ref{lem:Lopatinskii}, we show that the system \eqref{eq:bigsys} satisfies the Lopatinskii conditions when $\lambda$ is in the complement of the right side of the \eqref{eq: SpectrumCondition}.
Therefore, for such $\lambda$ and by \cite[Chapter 5, Proposition 11.16]{Taylor1}, when acting on $H^1(M)^{16}$ the corresponding operator is Fredholm. Considering that whenever \eqref{eq:Mluf} is satisfied we have \eqref{eq:bigsys}, we therefore conclude that $M(\lambda)$ is also Fredholm in this case. Thus $\lambda \in \sigma_{ess}(M)^c = \sigma_{ess}(L)^c$ which shows the right inclusion for \eqref{eq: SpectrumCondition}.

All that remains is to show that when
\begin{equation}\label{eq:boundfail}
\lambda \in \left (\bigcup_{x\in \partial M} i|P_n\hat{g}_0'|\left [ -\sqrt{\max(0,N^2)},\sqrt{\max(0,N^2)}\right 
]\right ) \bigcap \left (\bigcup_{(x,\xi)\in M\times \mathbb{R}^3 \setminus\{0\}} \sigma_{pt}(x,\xi) \right )^c,
\end{equation}
$\lambda \in \sigma_{ess}(M)$. For this last step, it is necessary to use the details of the computation checking the Lopatinskii condition, and so it is also proven in Lemma \ref{lem:Lopatinskii}. Therefore, using Lemma \ref{lem:Lopatinskii} the proof is complete.

\end{proof}

The next lemma is the key technical step in the proof of Theorem \ref{thm:pp}.

\begin{lemma}\label{lem:Lopatinskii}
    Suppose that $\lambda$ is in the complement of the right side of \eqref{eq: SpectrumCondition}.
    Then the system \eqref{eq:bigsys} satisfies the Lopatinskii conditions. Furthermore, suppose \eqref{eq:boundfail}. Then $\lambda \in \sigma_{ess}(M)$.
\end{lemma}

\begin{proof}
Let the operator on the left side of \eqref{eq:bigsys} be labeled $\mathcal{M}(
\lambda)$. Also suppose we collect the relevant operators for the boundary conditions in one large matrix 
\begin{equation}\label{eq:B}
\mathcal{B} = 
\begin{pmatrix}
    n\times & 0 & 0 & 0 & 0 & 0 & 0 & 0\\
    0 & n^T & 0 & 0 & 0 & 0 & 0 & 0\\
    0 & 0 & 1 & 0 & 0 & 0 & 0 & 0\\
    0 & 0 & 0 & n\times & 0 & 0 & 0 & 0\\
    0 & 0 & 0 & 0 & n^T & 0 & 0 & 0\\
    0 & 0 & 0 & 0 & 0 & 1 & 0 & 0
\end{pmatrix}.
\end{equation}
The principal symbol of $\mathcal{M}(\lambda)$ is
\begin{equation}
\sigma_p(\mathcal{M}(\lambda)) = i \rho_0
    \begin{pmatrix}
        \frac{g_0^T}{c^2} \xi\times & \xi^T & 0 & 0 & 0 & 0 & 0 & 0 \\
        0 & \xi \times & \xi & 0 & 0 & 0 & 0 & 0\\
        \xi^T & 0 & 0 & 0 & 0 & 0 & 0 & 0\\
        0 & 0 & 0 & \frac{g_0^T}{c^2} \xi\times & \xi^T & 0 & 0 & 0\\
        0 & 0 & 0 & 0 & \xi\times & \xi & 0 & 0 \\
        0 & 0 & 0 & \xi^T & 0 & 0 & 0 & 0\\
        V(\lambda) \xi \times & 0 & 0 & -\xi \times & 0 & 0 & \xi & 0\\
        0 & 0 & 0 & \xi \times & 0 & 0 & 0 & \xi
    \end{pmatrix}.
\end{equation}
This can be shown to be invertible if $V(\lambda)$ is invertible when projected onto the space orthogonal to $\xi$. Indeed, let us define
\[
V_{\xi_\perp\xi_\perp}(\lambda) = P_\xi^\perp V(\lambda) P_\xi^\perp, \ V_{\xi\xi_\perp}(\lambda) = P_\xi V(\lambda) P_\xi^\perp
\]
where $P_{\xi}$ is the projection onto the span of $\xi$ and $P_\xi^\perp$ the projection onto the space orthogonal to $\xi$. The condition \eqref{eq:lamintell} is equivalent to invertibility of $\widetilde{V}_\xi (\lambda) = V_{\xi_\perp \xi_\perp}(\lambda) + P_\xi$ at all points $(x,\xi) \in M \times (\mathbb{R}^3 \setminus \{0\})$. In the sequel we will suppress the dependence on $\lambda$ to ease the notation. When it exists, the inverse of $\sigma_p(\mathcal{M})$ is given by
\begin{equation} \label{eq:Msyminv}
\sigma_p(\mathcal{M})^{-1} = -\frac{i}{\rho_0 |\xi|^2}
    \begin{pmatrix}
        0 & 0 & \xi & 0 & 0 & 0 & - \xi \times \widetilde{V}^{-1}_\xi & - \xi \times \widetilde{V}^{-1}_\xi\\
        \xi & - \xi \times & 0 & 0 & 0 & 0 & -\xi \frac{g_0'^T}{c^2} P_\xi^\perp \widetilde{V}^{-1}_{\xi} & -\xi \frac{g_0'^T}{c^2} P_\xi^\perp \widetilde{V}^{-1}_{\xi}\\
        0 & \xi^T & 0 & 0 & 0 & 0 & 0 & 0\\
        0 & 0 & 0 & 0 & 0 & \xi & 0 & -\xi\times\\
        0 & 0 & 0 & \xi & -\xi\times & 0 & 0 & -\xi \frac{g_0'^T}{c^2}P_\xi^\perp\\
        0 & 0 & 0 & 0 & \xi^T & 0 & 0 & 0\\
        0 & 0 & 0 & 0 & 0 & 0 & \xi^T (I - V_{\xi\xi_\perp}) \widetilde{V}^{-1}_\xi & -\xi^T V_{\xi\xi_\perp} \widetilde{V}^{-1}_\xi P_\xi^\perp\\
        0 & 0 & 0 & 0 & 0 & 0 & 0 & \xi^T
    \end{pmatrix}.
\end{equation}
Let us consider the Lopatinskii condition in boundary normal coordinates $(\widetilde{x},x^3)$ where we freeze all coefficients at the central point where the Euclidean metric is the identity and write $n$ for the inward pointing unit normal vector. Without loss of generality we assume the central point is the origin. The condition is that there is a unique non-zero bounded solution of the system
\begin{equation}\label{eq:LopODE}
\sigma_p(\mathcal{M})(\widetilde{\xi}+ n D_3) U = 0, \quad \mathcal{B}U = \eta
\end{equation}
for any non-zero real $\widetilde{\xi} \in \mathbb{R}^{3}$ orthogonal to $n$ and $\eta \in \mathbb{C}^{8}$. Assuming $\lambda \in \sigma_{pt}((\widetilde{x},x^3),n)^c$, the ODE \eqref{eq:LopODE} is equivalent to
\[
\frac{\dd U}{\dd x^3} = -\sigma_p(\mathcal{M})\left (\frac{n}{i}\right )^{-1}\sigma_p(\mathcal{M})(\widetilde{\xi})U
\]
and checking the condition amounts to analysing the eigenvalues and eigenvectors of the matrix on the right side of this equation. Let us label this matrix
\[
K =  -\sigma_p(\mathcal{M})\left (\frac{n}{i}\right )^{-1}\sigma_p(\mathcal{M})(\widetilde{\xi}).
\]
Note that, because of \eqref{eq:Msyminv}, when the ellipticity condition is satisfied at the boundary $K$ cannot have any eigenvalues with zero real part. Considerable calculation shows that the eigenvalues of $K$ are $\pm |\widetilde{\xi}|$ each with algebraic multiplicity $7$ and
\[
\begin{split}
\alpha_\pm & = i|\widetilde{\xi}|\Bigg ( n^TV_{nn_\perp} \widetilde{V}_n^{-1} \hat{\xi}+\hat{\xi}^T\widetilde{V}_n^{-1}V_{n_\perp n} n \\
& \hskip1in \mp \sqrt{(n^TV_{nn_\perp} \widetilde{V}_n^{-1} \hat{\xi}-\hat{\xi}^T\widetilde{V}_n^{-1}V_{n_\perp n} n)^2 - 4(\hat{\xi}\widetilde{V}_n^{-1}\hat{\xi})n^T(V_{nn}-V_{nn_\perp}\widetilde{V}_n^{-1}V_{n_\perp n})n} \Bigg ) /2
\end{split}
\]
with multiplicity 1, or possibly $\pm|\widetilde{\xi}|$ with multiplicity $8$ if $\alpha_\pm  = \pm |\widetilde{\xi}|$. Note that, provided \eqref{eq:lamintell} holds, $\alpha_\pm$ must have non-zero real part by the ellipticity condition.

Let us introduce the notation
\[
\hat{\xi} = \frac{\widetilde{\xi}}{|\widetilde{\xi}|}, \quad n_\perp = \hat{\xi} \times n.
\]
Eigenvectors for $\pm |\widetilde{\xi}|$ are
\[
U_{1,\pm} = \begin{pmatrix}
    n \pm i \hat{\xi}\\
    0\\0\\0\\0\\0\\0\\0
\end{pmatrix},\
U_{2,\pm} = \begin{pmatrix}
    0\\ n \pm i \hat{\xi}\\0\\0\\0\\0\\0\\0
\end{pmatrix},
U_{3,\pm} = \begin{pmatrix}
    0\\n_\perp\\\pm i \\0\\0\\0\\0\\0
\end{pmatrix},
\]
\[
U_{4,\pm} = \begin{pmatrix}
    0\\0\\0\\n\pm i\hat{\xi}\\0\\0\\0\\0
\end{pmatrix},\
U_{5,\pm} = \begin{pmatrix}
    0\\0\\0\\0\\n\pm i \hat{\xi}\\0\\0\\0
\end{pmatrix},\
U_{6,\pm} = \begin{pmatrix}
    0\\0\\0\\0\\n_\perp\\\pm i\\0\\0
\end{pmatrix},
\]
and there are either eigenvectors or generalised eigenvectors for $\pm|\widetilde{\xi}|$ of the form
\[
U_{7,\pm} = \begin{pmatrix}
    0\\0\\0\\n_\perp\\a_{7,\pm}n+ b_{7,\pm}\hat{\xi}\\0\\\pm i\\\mp i
\end{pmatrix}
\]
for some constants $a_{7,\pm}$, $b_{7,\pm} \in \mathbb{C}$. Finally, either eigenvectors for $\alpha_\pm$ or generalised eigenvectors for $\pm |\tilde{\xi}|$ are given by
\[
U_{8,\pm} = \begin{pmatrix}
    2(\hat{\xi}^T\widetilde{V}_n^{-1} \hat{\xi}) n_\perp + a_{8,
    \pm} n + b_{8,\pm}\hat{\xi}\\ c_{8,
    \pm} n + d_{8,\pm}\hat{\xi}\\0\\0\\0\\0\\(n^TV_{nn_\perp} \widetilde{V}_n^{-1} \hat{\xi}-\hat{\xi}^T\widetilde{V}_n^{-1}V_{n_\perp n} n) \hskip3.8in \\
    \hskip1in \pm \sqrt{(n^TV_{nn_\perp} \widetilde{V}_n^{-1} \hat{\xi}-\hat{\xi}^T\widetilde{V}_n^{-1}V_{n_\perp n} n)^2 - 4 (\hat{\xi}^T\widetilde{V}_n^{-1}\hat{\xi})n^T(V_{nn}-V_{nn_\perp}\widetilde{V}_n^{-1}V_{n_\perp n})n }\\0
\end{pmatrix}
\]
for some constants $a_{8,\pm}$, $b_{8,\pm}$, $c_{8,\pm}$, $d_{8,\pm} \in \mathbb{C}$. For the Lopatinskii condition we must restrict to the generalised eigenspace corresponding to eigenvalues with negative real part. Thus, existence of a unique bounded solution of \eqref{eq:LopODE} is equivalent to a unique solution $(a_1, \ ... \ , a_8) \in \mathbb{C}^8$ of the system
\[
\mathcal{B}\sum_{j=1}^8 a_j U_{j,-} = \eta.
\]
Using \eqref{eq:B} and the equations for $U_{j,-}$ above we see that this linear system will have a unique solution if and only if $\hat{\xi}^T \widetilde{V}^{-1}_n \hat{\xi} \neq 0$. Calculation shows
\[
\widetilde{V}_{n}^{-1} = P_n + \frac{1}{\lambda^4 + \lambda^2 (N^2 |P_n^\perp\hat{g}_0'|^2 + 4 \Omega_n^2)} \Big (\lambda^2 P_n^\perp -2\lambda \Omega_n R_n + N^2 |P_{n}^{\perp}\hat{g}_0'|^2 P_{n}^\perp P_{(P_{n}^{\perp}\hat{g}_0')}^\perp P_{n}^\perp  \Big ),
\]
and so, since $\hat{\xi}$ is orthogonal to $n$,
\[
\hat{\xi}^T \widetilde{V}^{-1}_n \hat{\xi} = \frac{\lambda^2 + N^2 |P_{n}^{\perp}\hat{g}_0'|^2 \hat{\xi}^T P_{(P_{n}^{\perp}\hat{g}_0')}^\perp\hat{\xi}}{\lambda^4 + \lambda^2 (N^2 |P_n^\perp\hat{g}_0'|^2 + 4 \Omega_n^2)}
\]
Therefore, for $\lambda$ satisfying the interior ellipticity condition \eqref{eq:lamintell}, the Lopatinskii condition fails if and only if
\[
\lambda^2 = -N^2 |P_{n}^{\perp}\hat{g}_0'|^2 \hat{\xi}^T P_{(P_{n}^{\perp}\hat{g}_0')}^\perp\hat{\xi}.
\]
If $|P_n^\perp g_0'| \neq 0$, then $\hat{\xi}^T P_{(P_{n}^{\perp}\hat{g}_0')}^\perp\hat{\xi}$ takes all values in $[0,1]$ while if $|P_n^\perp g_0'|=0$ then the right side of this equation is always equal zero. Therefore, we see that the range of possible values of $\lambda$ satisfying this equation is $|P_n^\perp \hat{g}_0'| [-\sqrt{-N^2},\sqrt{-N^2}]$. If $N^2<0$, this is already contained in the interior part of the essential spectrum given by the first line of \eqref{eq: SpectrumCondition}. If $N^2\geq 0$, this interval will not be contained in the interior part of the essential spectrum and is given, for a single $x \in \partial M$, by the second line in \eqref{eq: SpectrumCondition} (see Figure \ref{fig:pointwise}(a)).

It remains to show that given \eqref{eq:boundfail}, $\lambda \in \sigma_{ess}(M)$. We will do this by showing the existence of a Weyl sequence. Indeed, by the calculations above, we see that when the Lopatinskii condition fails, for some $\widetilde{\xi}$ orthogonal to $n$ if we set $\zeta = U_{8,-} - i b_{8,-} U_{1,-} - i d_{8,-} U_{2,-}$, then we have
\[
\mathcal{B}\zeta = 0.
\]
Since $\zeta$ is composed of eigenvectors for eigenvalues with negative real part, there will be a corresponding non-zero bounded solution $U_\zeta$ of the ODE in \eqref{eq:LopODE} with $U_{\zeta}(
\widetilde{\xi},x^3 = 0) = \zeta$. Given $\epsilon>0$, let us choose a neighborhood $\Omega$ of $x$ sufficiently small so that all coefficients of operator $\mathcal{M}$ vary by at most $\epsilon$ within the neighborhood, and let $\phi \in C_c^\infty(\Omega)$. Then we set
\begin{equation}\label{eq:Udef}
\mathcal{U}(x) = \phi(x) e^{i t \widetilde{x} \cdot \widetilde{\xi}} U_\zeta(\widetilde{\xi},tx^3)
\end{equation}
which is in $H^1(M)^{16}$. With this choice of $\mathcal{U}$ we have
\[
\begin{split}
\mathcal{M}(\lambda) \mathcal{U}(x) & = \mathcal{M}(\lambda)|_{x = 0} \mathcal{U} + \epsilon \mathcal{O}(t)\\
& = it \phi(x) \sigma_p(\mathcal{M})|_{x =0}(\widetilde{\xi} + n D_3)\mathcal{U} + \epsilon \mathcal{O}(t) + \mathcal{O}(1)\\
& = \epsilon \mathcal{O}(t) + \mathcal{O}(1)
\end{split}
\]
as $t \rightarrow \infty$ with norm $H^1(M)^{16}$. Now let $w_u$ and $z_u$ be the corresponding components of $\mathcal{U}$. Since $\hat{\xi}^T \tilde{V}_n^{-1} \hat{\xi} = 0$, in the case when $\alpha_- \neq -|\widetilde{\xi}|$ these are explicitly given by
\begin{equation} \label{eq:wzdef}
\begin{split}
w_u & = e^{t(x^3\alpha_-+ i\widetilde{x}\cdot \widetilde{\xi})}(a_{8,-} n + b_{8_,-}\hat{\xi}) -ib_{8,-} e^{t(-x^3|\widetilde{\xi}| + i\widetilde{x}\cdot \widetilde{\xi})}(n-i\hat{\xi}),\\
z_u & = e^{t(x^3\alpha_-+ i\widetilde{x}\cdot \widetilde{\xi})}(c_{8,-} n + d_{8_,-}\hat{\xi}) -id_{8,-} e^{t(-x^3|\widetilde{\xi}| + i\widetilde{x}\cdot \widetilde{\xi})}(n-i\hat{\xi}).
\end{split}
\end{equation}
In the case that $\alpha_- = -|\widetilde{\xi}|$ and $U_{8,-}$ is a generalized eigenvector, these are replaced by
\begin{equation}\label{eq:wzdefgen}
\begin{split}
w_u & = e^{t(-x^3|\widetilde{\xi}|+ i\widetilde{x}\cdot \widetilde{\xi})}(a_{8,-}-ib_{8,-}) n  +tx^3 e^{t(-x^3|\widetilde{\xi}| + i\widetilde{x}\cdot \widetilde{\xi})}(n-i\hat{\xi}),\\
z_u & = e^{t(-x^3|\widetilde{\xi}|+ i\widetilde{x}\cdot \widetilde{\xi})}(c_{8,-}-id_{8,-}) n  +tx^3 e^{t(-x^3|\widetilde{\xi}| + i\widetilde{x}\cdot \widetilde{\xi})}(n-i\hat{\xi}).
\end{split}
\end{equation}
Then, considering the first component of \eqref{eq:bigsys}, we have $\nabla \times (\rho_0 w_u) + \rho_0 z_u \in D(T)$ and 
\[
T (\nabla \times (\rho_0 w_u) + \rho_0 z_u) = \epsilon \mathcal{O}(t) + \mathcal{O}(1).
\]
By the construction of $\pi_2$ described just above Lemma \ref{lem:pi12}, we have
\[
\pi_2(\nabla \times (\rho_0 w_u) + \rho_0 z_u) = \nabla \times (\rho_0 w_u) + \rho_0 z_u + \epsilon \mathcal{O}(t) + \mathcal{O}(1)
\]
with the norm $H$. With this in mind, let us set $u = \pi_2(\nabla \times (\rho_0 w_u) + \rho_0 z_u) \in \operatorname{Ker}(T)$, and consider $M(\lambda)u$. Using the last and second to last lines in \eqref{eq:bigsys} and the fact that most components of $\mathcal{U}$ are zero, we obtain
\[
M(\lambda)u = \epsilon \mathcal{O}(t) + \mathcal{O}(1).
\]
To construct a Weyl sequence, we need to normalize $u$, and so we consider $\|\nabla \times (\rho_0 w_u) + \rho_0 z_u\|_H$. In the case $U_{8,-}$ is not a generalized eigenvector, using \eqref{eq:wzdef} we see that
\[
\nabla \times (\rho_0 w_u) + \rho_0 z_u = t  e^{t(x^3\alpha_-+ i\widetilde{x}\cdot \widetilde{\xi})}\left ( -\frac{\alpha_-}{|\widetilde{\xi}|} b_{8,-} + i a_{8,-} \right ) |\widetilde{\xi}| n_\perp +  t e^{t(-x^3|\widetilde{\xi}|+ i\widetilde{x}\cdot \widetilde{\xi})} i b_{8,-}|\widetilde{\xi}| n_\perp + \mathcal{O}(1). 
\]
Since, from the calculation constructing $U_{8,-}$, we know that $a_{8,-}$ and $b_{8,-}$ cannot simultaneously vanish, from this last formula we see that
\[
\|u\|_H = \|\nabla \times (\rho_0 w_u) + \rho_0 z_u\|_H + \epsilon \mathcal{O}(t) + \mathcal{O}(t) \approx \mathcal{O}(t).
\]
By this notation, we mean that $\|u\|_H$ is bounded below by $C t$ as $t\rightarrow \infty$ for some constant $C>0$. A similar calculation beginning with \eqref{eq:wzdefgen}, omitted here, proves the same result when $U_{8,-}$ is a generalized eigenvector.
Therefore
\[
M(\lambda) \frac{u}{\|u\|_H} = \epsilon\mathcal{O}(1) + \mathcal{O}(t^{-1})
\]
and so by choosing $t$ sufficiently large we can obtain a sequence $v_\epsilon = u/\|u\|_H \in \operatorname{Ker}(T)$ with $H$-norm equal to one and such that $M(\lambda)v_\epsilon \rightarrow 0$ as $\epsilon \rightarrow 0$. Because of the oscillatory nature of \eqref{eq:Udef}, it is also clear that $v_\epsilon$
converges weakly to zero, meaning it is a Weyl sequence and so $\lambda \in \sigma_{ess}(M)$. This completes the proof.
\end{proof}

\section{Full spectrum bound} \label{sec:right}

In section \ref{sec:pure-point-dense}, we completely characterised the essential spectrum of $L$. We are unable to do the same for the full spectrum, but we can constrain $\sigma(L)$ as in \cite[Theorem 1]{dyson1979perturbations}. For completeness, we include a proof of Proposition \ref{prop:DS}.

\begin{proposition}[Dyson and Schutz] \label{prop:DS}
The spectrum $\sigma(L)$ satisfies
\begin{enumerate}
\item
\[
   \sigma(L) \subseteq \ii \RR \cup \{ \lambda \in \C\ :\
            |\operatorname{Im}(\lambda)| \le |\Omega| \};
\]
\item
while $A_2$ is bounded below by $\gamma(A_2)$, $\lambda \in \sigma(L)$ and $\lambda \notin i \mathbb{R}$,
\[
   |\lambda|^2 \le \max(0,-\gamma(A_2)) .
\]
\end{enumerate}
\end{proposition}

\begin{proof}
    We begin with introducing the sets
\[
   \mathscr{R}_\Omega = \ii \mathbb{R}
                       \cup \{ \lambda \in \mathbb{C}\ :\
            |\operatorname{Im}(\lambda)| \leq \| R_\Omega \| \}
\]
and
\[
   \mathscr{S}(c) = \{ \lambda = \zeta + \ii \xi \in \mathbb{C}\ :\
       \zeta, \xi \in \mathbb{R} ,\ \zeta^2 - \xi^2 \leq c \} .
\] 
\textbf{Step 1}: \textit{rough bound for $\sigma(L)$}. Let us first assume that $A_2$ is bounded below by $\gamma(A_2)$. Setting $\lambda = \zeta + i \xi$ with $\zeta$, $\xi \in \mathbb{R}$ and taking $u \in D(A_2)$, we first estimate
\begin{multline}
   \| (\lambda^2 \Id + A_2) u \|_H^2
        = \| (\zeta^2 - \xi^2) u + A_2 u + 2 \ii \zeta \xi u \|_H^2
\\
   = \| (\zeta^2 - \xi^2) u + A_2 u \|_H^2 + \| 2 \ii \zeta \xi u \|_H^2
   \geq ((\zeta^2 - \xi^2 + \gamma(A_2))^2
             + 4 \zeta^2 \xi^2) \| u \|_H^2
\end{multline}
provided that $\zeta^2 - \xi^2 + \gamma(A_2) > 0$. Hence, in this part
of the complex plane, $\lambda^2 \Id + A_2$ is invertible. We write
\[
   c_1 = \max(0,-\gamma(A_2)) + 1
\]
and find that $(\lambda^2 \Id + A_2)^{-1}$ is a bounded operator for
$\lambda \in \mathscr{S}(c_1)^c$ with
\begin{equation} \label{eq:Invest1}
   \| (\lambda^2 \Id + A_2)^{-1} \| \leq 
      (( \zeta^2 - \xi^2 + \gamma(A_2))^2 + 4 \zeta^2 \xi^2)^{-1/2} .
\end{equation}
Now, for $\lambda \in \mathscr{S}(c_1)^c$ we have
\[
L(\lambda) = (\lambda^2 \Id + A_2)( \Id + (\lambda^2 \Id + A_2)^{-1} 2 \lambda R_\Omega),
\]
and using \eqref{eq:Invest1}
\begin{multline*}
   \| (\lambda^2 \Id + A_2)^{-1} 2 \lambda R_\Omega \|
   \leq 2 |\lambda| \, \| R_\Omega \|
       ((\zeta^2 - \xi^2 + \gamma(A_2))^2 + 4 \zeta^2 \xi^2)^{-1/2}
\\
   = 2 |\lambda| \, \| R_\Omega \|
     [|\lambda|^4 + 2 (\zeta^2 - \xi^2) \gamma(A_2)
                         + \gamma(A_2)^2]^{-1/2} .
\end{multline*}
If $|\lambda|^2 = \zeta^2 + \xi^2 > c_2
>0$ is such that
the right-hand side is less than $1/2$, then
\[
\Id +  (\lambda^2 \Id + A_2)^{-1} 2 \lambda R_\Omega
\]
is invertible and so
\[
L(\lambda)^{-1} = (\Id +  (\lambda^2 \Id + A_2)^{-1} 2 \lambda R_\Omega)^{-1} (\lambda^2 \Id + A_2)^{-1}
\]
is bounded. Therefore, if we set $c_3 = \max(c_1,c_2)$, then
$\mathscr{S}(c_3)^c \subset \rho(L)$. Consequently, we have $\sigma(L) \subset \mathscr{S}(c_3)$.

\textbf{Step 2}: \textit{proof of Proposition~\ref{prop:DS}, 1.} We
assume that $\lambda \in \partial\sigma(L) = \sigma(L) \cap
\overline{\rho(L)}$. Applying Lemma~\ref{lem:seq}, we generate a sequence
$\{\lambda_\ell\}_{\ell=1}^\infty \subset \rho(L)$ and displacement
vectors $\{ u^\ell \}_{\ell=1}^\infty$ such that
\[
   \lim_{\ell \to \infty} \lambda_\ell = \lambda ,\quad  
   \| u^\ell \|_H = 1 ,\quad
   \lim_{\ell \to \infty} \| L(\lambda_\ell) u^\ell \|_H = 0 .
\]
It follows that
\[
   \lim_{\ell \to \infty} (L(\lambda_\ell) u^\ell,u^\ell) = 0 .
\]
We define the quantities
\begin{equation}\label{eq:stq}
   s_\ell = \frac{1}{\ii} \, (R_\Omega u^\ell,u^\ell) ,\quad 
   \tau_\ell = (A_2 u^\ell,u^\ell) ,\quad
   q_\ell = (L(\lambda_\ell) u^\ell,u^\ell)
\end{equation}
with $s_\ell \in \mathbb{R}$, $\tau_\ell \in \mathbb{R}$ and
$\lim_{\ell \to \infty} q_\ell = 0$, while
\[
   \lambda_\ell^2 + 2 \ii s_\ell \lambda_\ell + \tau_\ell = q_\ell .
\]
Writing $\zeta_\ell = \operatorname{Re}(\lambda_\ell)$, $\xi_\ell =
\operatorname{Im}(\lambda_\ell)$ and taking the
imaginary part of both sides,
\[
   2 \zeta_\ell \left(\xi_\ell
          + \frac{1}{\ii} \, (R_\Omega u^\ell,u^\ell)\right)
                                       = \operatorname{Im}(q_\ell) .
\]
Because the right-hand side goes to zero as $\ell \to \infty$, we have
\begin{equation} \label{eq:proofDS-1}
   \lim_{\ell \to \infty} \min \left\{ |\zeta_\ell| , 
   \left|\xi_\ell + \frac{1}{\ii} \, (R_\Omega u^\ell,u^\ell)
         \right|\right\} = 0 .
\end{equation}
Clearly, for all $\ell$
\begin{equation}\label{eq:sbound}
   -\| R_\Omega \| \leq \frac{1}{\ii} (R_\Omega u_\ell,u_\ell)
   \leq \| R_\Omega \|
\end{equation}
and, hence, (\ref{eq:proofDS-1}) implies that
\[
   \lim_{\ell \to \infty}
        \operatorname{dist}(\lambda_\ell,\mathscr{R}_\Omega) = 0 .
\]
But then $\lambda \in \mathscr{R}_\Omega$. Therefore, $\partial \sigma(L)
= \sigma(L) \cap \overline{\rho(L)} \subset \mathscr{R}_\Omega$.

We will prove $\sigma(L) \subset \mathscr{R}_\Omega$ by
contradiction. Assume that $\sigma(L) \not\subset \mathscr{R}_\Omega$,
then there exists a $\xi_0 \in \mathbb{R}$ such that the line
$M(\xi_0) = \{\zeta + \ii \xi_0 \in \mathbb{C}\ :\ \zeta \in
\mathbb{R}\}$ parallel to the real axis intersects $\sigma(L) \setminus
\mathscr{R}_\Omega$. From the characterization of $\mathscr{R}_\Omega$
it follows that $|\xi_0| > \| R_\Omega \|$ and that $M(\xi_0) \cap
\rho(L) \ne \emptyset$ using the rough estimate $\sigma(L) \subset
\mathscr{S}(c_3)$. We define the set
\[
   \mathscr{T}(\xi_0) = \{\zeta \in \mathbb{R}\ :\
                    \zeta + \ii \xi_0 \in \sigma(L)\} .
\]
This set is a bounded and closed set in $\mathbb{R}$. We define
$\zeta_1 = \max \, \mathscr{T}(\xi_0)$, $\zeta_2 = \min \,
\mathscr{T}(\xi_0)$. Because $\sigma(L)$ is a closed set, $\lambda_1 =
\zeta_1 + \ii \xi_0$ and $\lambda_2 = \zeta_2 + \ii \xi_0$ belong to
$\sigma(L)$; from the definition of $\zeta_1$, $\zeta_2$ it follows that
$\lambda_1$ and $\lambda_2$ belong to $\partial \sigma$. However, we
proved that $\partial\sigma \subset \mathscr{R}_\Omega$ and that
$|\xi_0| > \| R_\Omega \|$ and, hence, $\lambda_1 = \lambda_2 = \ii
\xi_0$. This is a contradiction which completes the proof of part {\it 1} of Proposition \ref{prop:DS}.

\textbf{Step 3}: \textit{proof of Proposition~\ref{prop:DS}, 2.} We
assume that $\lambda \in \partial\sigma(L) = \sigma(L) \cap \overline{\rho(L)}$
with $\operatorname{Re}(\lambda) \ne 0$. Applying Lemma~\ref{lem:seq},
we generate a sequence $\{\lambda_\ell\}_{\ell=1}^\infty \subset \rho(L)$
and displacement vectors $\{u^\ell\}_{\ell=1}^\infty$ such that
\[
   \lim_{\ell \to \infty} \lambda_\ell = \lambda ,\quad  
   \| u^\ell \| = 1 ,\quad
   \lim_{\ell \to \infty} \| L(\lambda_\ell) u^\ell \| = 0
\]
as before. Also, let $s_\ell$, $\tau_\ell$ and $q_\ell$ be as in \eqref{eq:stq}. By \eqref{eq:sbound}, $s_\ell$ is bounded and since $q_\ell \rightarrow 0$ and $\lambda_\ell \rightarrow \lambda$, $\tau_\ell$ must also be bounded. Therefore $s_\ell$ and $\tau_\ell$ have convergent subsequences and, by passing to a subsequence if necessary, we can assume without loss of generality that $s_\ell$ and $\tau_\ell$ converge respectively to some $s$ and $\tau \in \mathbb{R}$. Then, taking the limit in \eqref{eq:stq} we obtain
\[
\lambda^2 + 2 i s \lambda + \tau = 0,
\]
which implies
\[
\lambda = -is \pm \sqrt{-s^2 - \tau}.
\]
If $-s^2 - \tau \leq 0$, then $\lambda \in i \mathbb{R}$ which we have excluded by assumption. Therefore, $s^2 + \tau < 0$ and
\[
|\lambda|^2 = s^2 + (-s^2 - \tau) = -\tau \leq -\gamma(A_2)
\]
This proves that if $\lambda \in \partial\sigma(L) \setminus \ii \mathbb{R}$, then $|\lambda|^2 \leq
\max(0,-\gamma(A_2))$.

% We now use that
% \[
%    \tau_\ell \in \mathbb{R} ,\quad \tau_\ell \geq \gamma(A_2) .
% \]
% The roots of $\lambda_\ell^2 + 2 \ii s_\ell \lambda_\ell + \tau_\ell =
% q_\ell$ are
% \[
%    \lambda_\ell = -\ii s_\ell
%                 \pm (-s_\ell^2 - \tau_\ell + q_\ell)^{1/2} .
% \]
% Since we have assumed that $\operatorname{Re}(\lambda) \ne 0$ and
% $\lim_{\ell \to \infty} \lambda_\ell = \lambda$, it follows that
% $-s_\ell^2 - \tau_\ell + q_\ell > 0$ for sufficiently large $\ell$.

% If $\gamma(A_2) \geq 0$, $\operatorname{Re}\lambda \ne 0$ which
% is excluded and, hence, we can assume that $\gamma(A_2) < 0$
% hereafter. We have
% \[
%    |\lambda_\ell|^2 = s_\ell^2 + (-s_\ell^2 - \tau_\ell + q_\ell)
%             = -\tau_\ell + q_\ell \leq -\gamma(A_2) + q_\ell .
% \]
% Taking $\ell \to \infty$, we get 
% \[
%    |\lambda|^2 = \limsup_{\ell \to \infty} |\lambda_\ell|^2
%                \leq -\gamma(A_2) .
% \]
% We proved that $\lambda \in \partial\sigma \setminus \ii \mathbb{R}$,
% $\gamma(A_2) < 0$ implies that $|\lambda|^2 \leq -\gamma(A_2) =
% \max(0,-\gamma(A_2))$.

Next, we prove that $\lambda \in \sigma(L) \setminus \ii \mathbb{R}$
implies that $|\lambda|^2 \leq \max(0,-\gamma(A_2))$. We introduce
\[
   \mathscr{R}_\Omega^\prime = \{\lambda \in \mathbb{C}\ :\
     |\operatorname{Im}\lambda| \leq \| R_\Omega \| ,\
     |\lambda|^2\leq \max(0,-\gamma(A_2))\} .
\]
We already know that $\partial\sigma(L) \setminus \ii \mathbb{R} \subset
\mathscr{R}_\Omega^\prime$. Now, we assume that
\begin{equation} \label{eq:proofDS-2}
   (\sigma(L) \setminus \ii \mathbb{R})
                     \not\subset \mathscr{R}_\Omega^\prime,
\end{equation}
and let us use the same notation $M(\xi_0)$ and $\mathscr{T}(\xi_0)$ as in {\it Step 2} above.
Then there exists a $\xi_0 \in [-\| R_\Omega \|,\| R_\Omega \|]$ such
that the line $M(\xi_0)$ parallel to the real axis intersects $(\sigma(L)
\setminus \ii \mathbb{R}) \setminus \mathscr{R}_\Omega^\prime$. With
this $\xi_0$, $\mathscr{T}(\xi_0)$ is a bounded and closed set using
the result obtained in \textit{Step 1}. We define, again, $\zeta_1 =
\max \, \mathscr{T}(\xi_0)$, $\zeta_2 = \min \,
\mathscr{T}(\xi_0)$. Then $\lambda_1 = \zeta_1 + \ii \xi_0$,
$\lambda_2 = \zeta_2 + \ii \xi_0$ necessarily belong to
$\partial\sigma(L)$. Due to assumption (\ref{eq:proofDS-2}), we have
$|\lambda_1| > \max(0,-\gamma(A_2)$ or $|\lambda_2| >
\max(0,-\gamma(A_2))$. This is a contradiction to $\partial\sigma(L)
\subset \mathscr{R}_\Omega^\prime$.
\end{proof}

The next lemma was used in the proof of Proposition \ref{prop:DS}.

\begin{lemma}\label{lem:seq}
    Suppose that $\lambda \in \partial \sigma(L)$. Then there exists a sequence $\lambda_\ell \in \rho(L)$ and $u^\ell \in D(L)$ such that
\[
   \lim_{\ell \to \infty} \lambda_\ell = \lambda ,\quad  
   \| u^\ell \|_H = 1 ,\quad
   \lim_{\ell \to \infty} \| L(\lambda_\ell) u^\ell \|_H = 0 .
\]
\end{lemma}

\begin{proof}
    Let $\lambda \in \partial \sigma(L)$ and so there exists a sequence $\lambda_\ell \in \rho(L)$ such that $\lambda_\ell \rightarrow \lambda$. Since $\sigma(L)$ is closed, $\lambda \in \sigma(L)$. Suppose that $L(\lambda)$ is not injective. Then, the lemma is proven by taking $u_\ell$ a constant sequence with norm $1$ in the kernel of $L(\lambda)$. So, assume now that $L(\lambda)$ is injective. Then $L^{-1}(\lambda)$ is defined on some domain in $H$, and if this domain is all of $H$ then $L(\lambda)^{-1}$ must be continuous. Thus, there must exist $f \in H$ which is not in the range of $L(\lambda)$. Define
    \[
    v_\ell = L(\lambda_\ell)^{-1} f.
    \]
    We then claim that some subsequence of $u_\ell = v_\ell/\|v_\ell\|_H$ is a sequence of unit vectors that satisfies
    \begin{equation}\label{eq:decay}
    \lim_{\ell\rightarrow \infty}\|L(\lambda_\ell) u^\ell \|_H = 0.
    \end{equation}
    We will argue by contradiction.

    Indeed, suppose that no subsequence of $u_\ell$ as defined above satisfies \eqref{eq:decay}. Then
    \[
    \|L(\lambda_\ell) u^\ell\|_{H} \geq C > 0 \Rightarrow \|f\|_H =  \|L(\lambda_\ell) v^\ell\|_H \geq C \|v^\ell\|_{H}
    \]
    for some constant $C$ and all $\ell$. Therefore, since $v_\ell$ is bounded, it must have a weakly convergent subsequence: say $v_\ell$ converges to $v$ weakly.

    Now suppose $g \in D(A_2)$. Then
    \[
    \langle f, g \rangle_H = \langle L(\lambda_\ell) v^\ell, g \rangle_H = \langle v^\ell, L(\lambda_\ell)^* g \rangle_H = \langle v^\ell, (L(\lambda_\ell)^* - L(\lambda)^*) g \rangle_H + \langle v^\ell, L(\lambda)^*g \rangle_H
    \]
    Taking the limit as $\ell \rightarrow \infty$ gives
    \[
    \langle f, g \rangle_H = \langle v, L(\lambda)^*g \rangle_H.
    \]
    Therefore $v \in D(L)$ and $L(\lambda)v = f$. This is a contradiction since $f$ was assumed to be outside the range of $L$. This completes the proof.
\end{proof}

\noindent
If $\gamma(A_2) \geq 0$, it follows immediately from Proposition \ref{prop:DS} that $\sigma \subseteq \ii \RR$, but this is unlikely to be the case. In general, Proposition \ref{prop:DS} provides an upper bound on the full spectrum $\sigma(L)$ which is illustrated in Figure \ref{fig:cross}.

% We apply this theorem to $A_{2;22}$ {\color{red} [@Sean: could you please complete the proof .. some notes: We obtain the superset $\mathfrak{S}_2$ for the spectrum, $\sigma_2$, as follows. For the part of $\mathfrak{S}_2$ away from the imaginary axis, $|\omega| \le |\Omega|\ \text{and}\ \omega^2 + \nu^2 \le \max(0, -\inf_{u \in H_2 , \| u \| = 1} (A_{2;22} u,u))$, we apply
% Proposition~\ref{prop:DS}. Perhaps to constrain the part of
%     $\sigma_2$ on the imaginary axis, $\nu = 0\ \text{and}\ |\omega|
%     \le |\Omega| + [\Omega^2 + \max(0, \sup_{u \in H_2 , \| u \| = 1}
%       (A_{2;22} u,u))]^{1/2}$: Use that if $-\ii \omega \in \sigma_2$
%     then $|\omega| \le |p^{\pm}(u)|$ for all $u \in H_2$].
% The superset $\mathfrak{S}_3$ for the essential spectrum,
% $\sigma_{2;\mathrm{ess}}$, follows the same method of proof upon
% replacing $A_{2;22}$ by $S_1$.}

\section{Discussion}

We precisely characterized the spectrum of rotating truncated gas planets for both variable \tred{densities and variable} positive and negative \tred{squared} Brunt-V\"{a}is\"{a}l\"{a} frequencies. Acoustic modes correspond with part of the point spectrum (while other modes such as quasi-rigid body modes are also associated with the point spectrum) and inertia-gravity modes with the essential spectrum. We presented a partial resolution of the identity with acoustic modes which reveals inaccuracies in common approaches to compute these.

A further study of the dynamics and attractors associated with the inertia-gravity modes described in this paper will be left for future research. We note that such analysis was carried out by Colin de Verdi\`{e}re and Vidal \cite{VidalCdV:2024}. In preparation for this, making the connection to their work explicit, we briefly relate our formulation to theirs. We introduced
\begin{equation}
   \tilde{s} = \nabla \rho_0 - \frac{\rho_0}{c^2} g_0'
\end{equation}
and identified the dynamic pressure as
\begin{equation}
   P = -c^2 [\nabla \cdot (\rho_0 u) - \tilde{s} \cdot u]
\end{equation}
or
\begin{equation}
   P = -\rho_0\ [c^2 \nabla \cdot u + g_0' \cdot u] .
\end{equation}
Using that
\begin{equation}
   \tilde{s} \cdot u = \frac{\tilde{s} \cdot g_0'}{|g_0'|^2}
   (g_0' \cdot u) = \frac{N^2}{|g_0'|^2} (g_0' \cdot (\rho_0 u))
\end{equation}
as $\nabla \rho_0$ and $g_0'$ must be parallel, we obtain
\begin{equation}
   P = -c^2 \left[ \nabla \cdot (\rho_0 u)
   - \frac{N^2}{|g_0'|^2} (g_0' \cdot (\rho_0 u)) \right] .
\end{equation}
While introducing the particle velocity, $v = \partial_t u$, equations (\ref{eq: MCts}) and (\ref{eq: PoissonEqPert-2}) are equivalent to the system
\begin{eqnarray}
   \partial_t \rho + \nabla \cdot (\rho_0 v) &=& 0 ,
\\[0.25cm]   
   \partial_t (\rho_0 v) + 2 \Omega \times (\rho_0 v) &=&
   -\nabla P + \rho g_0' - \rho_0 \nabla \Phi' ,
\label{eq:m}
\\
   \partial_t P &=& c^2 \left[ \partial_t \rho
   + \frac{N^2}{|g_0'|^2} (g_0' \cdot (\rho_0 v)) \right] ,
\end{eqnarray}
supplemented with (\ref{eq: PoissonEqPert-2}), which is equivalent to the system in linearized hydrodynamics as in \cite{Prat-etal:2016}
\begin{eqnarray}
   \partial_t \rho + \nabla \cdot (\rho_0 v) &=& 0 ,
\\[0.25cm]   
   \partial_t (\rho_0 v) + 2 \Omega \times (\rho_0 v) &=&
   -\nabla P + \rho g_0' - \rho_0 \nabla \Phi' ,
\\[0.25cm]
   \partial_t P + v \cdot \nabla P_0 &=&
   c^2 [ \partial_t \rho + v \cdot \nabla \rho_0 ]
\end{eqnarray}
as $\nabla P_0 = -\rho_0 g_0'$. In the Cowling approximation, one drops the term $-\rho_0 \nabla \Phi'$. If $u \in \ker(T)$ then $P = 0$ and
\begin{equation}
   \rho g_0' = -(\nabla \cdot (\rho_0 u)) g_0'
   = -(\tilde{s} \cdot u) g_0'
   = -N^2 \hat{g}_0' (\hat{g}_0' \cdot \rho_0 u) . 
\end{equation}
Then, \eqref{eq:m} is seen to be equivalent to
\[
\partial_t v + 2 \Omega \times  v + N^2 \hat{g}_0' (\hat{g}_0' \cdot u) = 0, \quad Tu = 0.
\]
which is closely related to \eqref{eq:igm_1}. %Representing $v$ in an orthonormal basis with third component $\hat{g}_0'$, this last equation can be shown to reduce to the system
%\begin{equation} \label{eq:red_1}
%   (\partial_t
%   + A) \left(\begin{array}{c}
%   v_{\perp,1} \\ v_{\perp,2} \\ v_{\parallel} \\ \rho' \end{array}\right)
%   = 0\quad\text{with}\
%   A = \left(\begin{array}{cccc}
%   0 & -2 \Omega_{\parallel} & 2 \Omega_{\perp,2} & 0 \\
%   2 \Omega_{\parallel} & 0 & -2 \Omega_{\perp,1} & 0 \\
%   -2 \Omega_{\perp,2} & 2 \Omega_{\perp,1} & 0 & N \\
%   0 & 0 & -N & 0 \end{array}\right) ,\quad T v = 0 ,
%\end{equation}
Upon first introducing
\begin{equation} \label{eq:red_2}
   \rho' = N (\underbrace{\hat{g}_0' \cdot u}_{u_{\parallel}}),
\end{equation}
this equation can be written as the system
\begin{equation} \label{eq:red_1}
   (\partial_t
   + A) \left(\begin{array}{c}
   v\\ \rho' \end{array}\right)
   = 0\quad\text{with}\
   A = \left(\begin{array}{cc}
   2 \Omega \times & N \hat{g}_0' \\
   -N \hat{g}_0'^T & 0
   \end{array}\right) ,\quad T v = 0.
\end{equation}
In \cite{VidalCdV:2024}, this system is formed by expressing $v$ in an orthogonal basis where one of the basis vectors is $\hat{g}_0'$. Including the projectors,
\begin{equation}
    \pi_2' \left(\begin{array}{c} v \\ \rho' \end{array}\right)
    = \left(\begin{array}{c} \pi_2 v \\ \rho' \end{array}\right) ,
\end{equation}
the system takes the form
\begin{equation}
   (\partial_t + H) \left(\begin{array}{c}
   v\\ \rho' \end{array}\right)
   = 0\quad\text{with}\
   H = \pi_2' A \pi_2'
\end{equation}
as in Colin de Verdi\`{e}re and Vidal \cite{VidalCdV:2024}, who considered the case when $\hat{g}_0'$ and $N$ are constants. (These authors consider the further spectral analysis of this equation which is, in turn related to the work of \cite{dyatlov2021mathematics} in case the (compact) manifold would not have a boundary.) The system needs to be supplemented with the boundary condition $u\cdot n|_{\partial M} = 0$ (see \eqref{eq:Tdef}). $H$ is identified with the Poincar\'{e} operator. The spectrum of $H$ is $\sigma_{ess}(L_{22})$.

We can write the constrained system (\ref{eq:red_1}) (in the Cowling approximation) in the form
\begin{equation}
   \left(\begin{array}{cc}
   -\operatorname{i} \lambda + \pi'_2 A &
   \begin{array}{c} \nabla \\ 0 \end{array} \\
   \begin{array}{cc} T & 0 \end{array} & 0
   \end{array}\right) 
   \left(\begin{array}{c} v \\ \rho' \\ P \end{array}\right)
   = 0.
\end{equation}
The principal, $\sigma_p(\pi_2)$ corresponds with the Leray projector and is given by \eqref{eq:pisymbol}. Consistent with the Leray projector, we may restrict $u \in \ker T_1$, $T_1 u = \rho_0 \nabla \cdot u$ when $P = P_1 = \rho_0 g_0' \cdot u$ and is non-vanishing. Then
\begin{equation}
   \left(\begin{array}{cc}
   -\operatorname{i} \lambda + \sigma_p(\pi_2') A &
   \begin{array}{c} \nabla \\ 0 \end{array} \\
   \begin{array}{cc} \nabla\cdot & 0 \end{array} & 0
   \end{array}\right) 
   \left(\begin{array}{c} v \\ \rho' \\ P_1 \end{array}\right)
   = 0 .
\end{equation}
Keeping the principal parts, eliminating $v$ and $\rho'$, leads to a Poincar\'{e} equation for $P_1$, $S_{\omega} P_1 = 0$, where the principal symbol of $S_{\omega}$ is given by
\[
\begin{split}
   s_{\omega}(x,\xi) & = \det\left(\begin{array}{cc}
   -\operatorname{i} \lambda + \sigma_p(\pi_2')(\xi) A(x) &
   \begin{array}{c} \operatorname{i} \xi \\ 0 \end{array} \\
   \begin{array}{cc} \operatorname{i} \xi^T & 0 \end{array} & 0
   \end{array}\right)\\
   & = -i \lambda |\xi|^2 (\lambda^2 + 4 \Omega_\xi^2 + N^2 |P_{\xi}^{\perp} \hat{g}_0'|^2).
\end{split}
\]
This determinant can be calculated by considering the matrix in an orthonormal basis that includes $(\xi/|\xi|,0,0)^T)$ as one of the basis vectors. Therefore, $S_\omega$ is elliptic except when $\lambda = 0$ or
\[
\lambda^2 = - 2 \Omega_\xi^2 - N^2 |P_\xi^\perp \hat{g}_0'|^2 ,
\]
which corresponds with Lemma~\ref{lem:pt}.

Future work includes a generalization to the precise characterization of the spectra of rotating terrestrial planets involving boundary conditions at the core-mantle \tred{interface} different from the ones appearing in the present results, and extending the work of Valette \cite{Valette:1989a}. It also includes removing the truncation employed in the present analysis of gas planets by letting $c^2$ vanish (proportional to the pseudo-enthalpy in a polytropic model) and $N^2$ blow-up (proportional to $c^{-2}$ in a polytropic model) at the boundary, see Prat \textit{et al.} \cite{Prat-etal:2016}.

% Then
% \[
%     P_0 = K \rho_0^{1 + \tfrac{1}{\mu}}
% \]
% while $\rho_0$ tends to zero {\color{red} [at which rate? what about its gradient]} at the boundary. Introduce the pseudo-enthalpy,
% \[
%     h_0 = \int \operatorname{d}P_0 / \rho_0
%     = (1 + \mu) \frac{P_0}{\rho_0}
%     = K (1 + \mu) \rho_0^{\tfrac{1}{\mu}} .
% \]
% The sound speed becomes
% \[
%     c^2 = \frac{\Gamma_1}{\mu + 1} h_0
% $\]
% and the Brunt-V\"{a}is\"{a}l\"{a} frequency becomes
% \[
%     N^2 = \left( \frac{\mu \Gamma_1}{\mu+1} - 1 \right)
%     \frac{|g_0|^2}{c^2} ,\quad g_0 = \nabla h_0 . 
% \]
% The first adiabatic exponent of the gas is $\Gamma_1 = 5/3$. It is clear that $c^2$ tends to zero at the boundary and that $N^2$ blows up at the boundary. 

\section{Acknowledgements}

\tred{We would like to thank the anonymous referees for comments and suggestions which improved the manuscript. We would also like to thank Yves Colin de Verdi\`{e}re, Bernard Vallette, Mikko Salo and Marco Marletta for useful conversations during this research}

MVdH was supported by the Simons Foundation under the MATH + X program and the National Science Foundation under grant DMS-2407456. Contributions from SH were supported by the Engineering and Physical Sciences Research Council (EPSRC) through grant EP/V007742/1. The authors would like to thank the Isaac Newton Institute for Mathematical Sciences  for support and hospitality through the program Rich and Non-linear Tomography: A Multidisciplinary Approach, where work on this paper was undertaken. This program was supported by EPSRC grant number EP/R014604/1.  

\appendix

\section{A study of $\operatorname{Ker}(A_2)$}
\label{A:r}

In this appendix, we study $\operatorname{Ker}(A_2)$ which leads to
the introduction of rigid body motions with corresponding eigenvalues not contained in $\sigma_1$. \tred{Here we review modes associated with these motions and provide explicit calculations.}

Quasi-rigid displacements are those with zero strain. The vanishing strain condition implies that they must be of the form
\begin{equation} \label{eq:qrd}
   u(x) = t + k \wedge x ,\quad t, k \in \C^3 .
\end{equation}
We briefly demonstrate that such quasi-rigid displacements are zero
eigenfunctions in the case of a non-rotating planet. In a rotating planet the quasi-rigid motions are still eigenfunctions, but they have different eigenvalues as we will demonstrate below.

\subsection{Non-rotating planet}

To consider the non-rotating case we first calculate
\[
   u(x) \cdot \nabla \nabla \Phi^0(x)
      = G \lim_{\epsilon \rightarrow 0^+}
        \int_{\tred{M} \setminus B_\epsilon (x)}
           \frac{x - x'}{|x - x'|^3}\cdot u(x)\ \nabla \rho^0 (x')\
                        \dd x',
\]
\[
   \nabla S(u)(x) = -G \lim_{\epsilon \rightarrow 0^+}
        \int_{\tred{M} \setminus B_\epsilon (x)} \nabla u(x')
           \frac{x - x'}{|x - x'|^3}\ \rho^0(x')
        + \frac{x - x'}{|x - x'|^3}\cdot u(x')\ \nabla \rho^0(x')\
                        \dd x',
\]
and
\[
   \nabla \Phi^0(x)^T \nabla u(x)
       = G \lim_{\epsilon \rightarrow 0^+}
         \int_{\tred{M} \setminus B_\epsilon(x)}
      \frac{(x - x')^T}{|x - x'|^3} \nabla u(x)\ \rho^0(x') \dd x' .
\]
Using these formulae and the facts that $u$ given by \eqref{eq:qrd}
has zero strain, constant antisymmetric first differential (that is,
$\nabla u$ is a constant antisymmetric matrix), and vanishing second
derivatives we find that for such $u$
\[
\begin{split}
   A_2^0 u(x) & = \nabla S(u)(x) - \nabla \Phi^0(x)^T \nabla u(x)
       + u \cdot \nabla \nabla \Phi^0(x)
\\
   & = G \lim_{\epsilon \rightarrow 0^+}
       \int_{\tred{M} \setminus B_\epsilon(x)}
       \frac{(x - x')^T}{|x - x'|^3} (\nabla u(x') - \nabla u(x))\
                        \rho^0(x')
\\
   &\hskip1.5in + \frac{x - x'}{|x - x'|^3} \cdot (u(x) -  u(x'))\
                        \nabla \rho^0(x') \dd x'
\\
   & = G \lim_{\epsilon \rightarrow 0^+}
       \int_{\tred{M} \setminus B_\epsilon(x)}
       \frac{x - x'}{|x - x'|^3} \cdot (k \wedge (x -  x'))\
       \nabla \rho^0(x') \dd x'
\\
   & = 0
\end{split}
\]
This calculation shows that, as stated above, in the non-rotating case the quasi-rigid motions are eigenfunctions associated with the zero eigenvalue.

\subsection{Rotating planet}

Now we consider the rotating case. First, from the above calculation, since $A_2^0 u = 0$ for the quasi-rigid motions, we find
\[
\begin{split}
   A_2 u(x) &= u(x) \cdot \nabla \nabla \Psi^s(x)
                       - \nabla \Psi^s(x)^T \nabla u(x)
\\
   & = (\Omega \cdot u)\ \Omega - \Omega^2\ u
       + (\Omega \cdot x)\ \Omega^T \nabla u - \Omega^2\ x^T\nabla u
\\
   & = (\Omega \cdot t)\ \Omega + \Omega \cdot (k \wedge x)\ \Omega
       - \Omega^2\ t - \Omega^2 k \wedge x
       - (\Omega \cdot x)\ k \wedge \Omega + \Omega^2 \ k \wedge x
\\
   & = (\Omega \cdot t)\ \Omega + \Omega \cdot (k \wedge x)\ \Omega
       - \Omega^2\ t - (\Omega \cdot x)\ k \wedge \Omega .
\end{split}
\]
From this formula we can immediately see that when $t = c_t \Omega$ and $k = c_k \Omega$,
\[
   A_2 u = c_t\ \Omega^2\ \Omega  - c_t\ \Omega^2 \ \Omega = 0 .
\]
Thus the quasi-rigid motion
\[
   u(x) = c_t\ \Omega + c_k\ \Omega \wedge x ,
\]
is an eigenfunction associated with zero eigenvalue in the rotating
case. The first term on the right-hand side is identified with the
\textit{axial translational mode}, while the second term on the right-hand side is identified with the \textit{axial spin mode}.

The other rigid motions correspond with non-zero eigenvalues. We must also include the first and second order terms. Thus we calculate for quasi-rigid motions:
\[
\begin{split}
   L(\lambda) u & = (\lambda^2 \Id + 2 \lambda R_{\Omega} + A_2) u
\\
   & = \lambda^2 t + \lambda^2 k \wedge x
         + 2 \lambda \Omega \wedge t
       + 2 \lambda \Omega \wedge (k \wedge x)
         + (\Omega \cdot t) \ \Omega
         + \Omega \cdot (k \wedge x)\ \Omega - \Omega^2\ t
         - (\Omega \cdot x)\ k \wedge \Omega .
\end{split}
\]
First, let us consider the case when $t \perp \Omega$ and $k = 0$. Then we have
\[
   L(\lambda)u = (\lambda^2 - \Omega^2) t
                         + 2 \lambda \Omega \wedge t .
\]
Choosing any $a \perp \Omega$ and setting
\[
   t_\pm = a \pm \ii \frac{\Omega \wedge a}{|\Omega|} m
\]
we have
\[
   \Omega \wedge t_\pm = \mp \ii | \Omega | t_\pm .
\]
Using this we see that with $\lambda_\pm = \pm \ii |\Omega |$
\[
\begin{split}
   L(\lambda_\pm)t_\pm & = - 2 \Omega^2 t_\pm
        + 2 (\pm \ii |\Omega|) (\mp \ii | \Omega | t_\pm) .
\end{split}
\]
Therefore, we have a two-dimensional space of eigenfunction and eigenvalue pairs given by
\[
   \lambda_\pm = \pm \ii |\Omega|, \quad t_\pm
      = a \pm \ii \frac{\Omega \wedge a}{|\Omega|}, \quad k = 0 .
\]
These are the so-called \textit{equatorial translation modes}. We note that the translation modes are not contained in $H$; thus the only mode in $\operatorname{Ker} A_2$ playing a role is the axial spin mode.

Now let us consider the case $t = 0$, and $k \perp \Omega$. Choosing
any $a \perp \Omega$ we have as before
\[
   k_\pm = a \pm \ii \frac{\Omega \wedge a}{|\Omega|}
       \Rightarrow \Omega \wedge k_\pm = \mp \ii | \Omega | k_\pm .
\]
Again, taking $\lambda_\pm = \pm \ii |\Omega|$, we find that
\[
\begin{split}
   L(\lambda_\pm) k_\pm \wedge x & = - \Omega^2  k_\pm \wedge x
      + 2 (\pm \ii |\Omega|) (\Omega \cdot x)  k_\pm 
      + \Omega \cdot \left( k_\pm \wedge x \right) \Omega
        - (\Omega \cdot x) \left ( k_\pm \wedge \Omega \right)
\\
   & = - \Omega^2 k_\pm \wedge x \pm 2 \ii |\Omega|
        (\Omega \cdot x)  k_\pm + \Omega^2 k_\pm \wedge x
                   + 2 (x \cdot \Omega) \Omega \wedge k_{\pm}
\\
   & = 0 .
\end{split}
\]
Therefore there is also a two-dimensional space of eigenfunction and
eigenvalue pairs given by
\[
   \lambda_\pm = \pm \ii |\Omega|, \quad t = 0 ,\quad
         k_\pm = a \pm \ii \frac{\Omega \wedge a}{|\Omega|} .
\]
These are the so-called tilt-over modes. \tred{We note that  tilt-over modes are well-known (according to Chandrasekhar \cite{Chandrasekhar:1969}) transversely skewed oscillations that can occur in rotating, precessing bodies}.

These calculations show that in moving from the non-rotating model to the rotating one the six-dimensional eigenspace of quasi-rigid modes with eigenvalue zero is split into three separate two-dimensional eigenspaces with eigenvalues $\pm \ii |\Omega|$ and
$0$.

\bibliographystyle{abbrv} 
\bibliography{GlobalSpectrum.bib}

\end{document}